\documentclass[preprint,12pt]{elsarticle}

\usepackage{amsmath,amssymb,amsfonts}
\def\BibTeX{{\rm B\kern-.05em{\sc i\kern-.025em b}\kern-.08em
    T\kern-.1667em\lower.7ex\hbox{E}\kern-.125emX}}

\graphicspath{{fig/}} 
\usepackage{subcaption}
\usepackage{dblfloatfix}

\usepackage{amsthm}  
\usepackage{mathtools}  
\usepackage{url}
\usepackage{comment}
\usepackage{siunitx}
\usepackage{color}

\newcommand{\norm}[1]{\lVert #1 \rVert}
\newcommand{\Norm}[1]{\left\lVert #1 \right\rVert}
\newcommand{\abs}[1]{\lvert #1 \rvert}

\DeclareMathOperator*{\minimize}{\ \ \ minimize\ \ \ }
\DeclareMathOperator*{\subjectto}{\ \ \ subject\ to\ \ \ }
\DeclareMathOperator{\dist}{dist}
\DeclareMathOperator{\out}{out}
\DeclareMathOperator{\stc}{stc}
\DeclareMathOperator{\dyn}{dyn}
\DeclareMathOperator{\orig}{orig}
\DeclareMathOperator{\reference}{ref}

\DeclareSymbolFont{bbold}{U}{bbold}{m}{n}
\DeclareSymbolFontAlphabet{\mathbbold}{bbold}

\usepackage{scalerel}

\usepackage{bm}
\usepackage[normalem]{ulem}

\newcommand{\add}[1]{#1}
\newcommand{\del}[1]{}

\usepackage{cleveref}

\crefname{equation}{}{}
\crefmultiformat{equation}{(#2#1#3)}{, (#2#1#3)}{, (#2#1#3)}{,  (#2#1#3)}

\theoremstyle{definition}
\newtheorem{definition}{Definition}[section]

\newtheorem{lemma}[definition]{Lemma}
\newtheorem{proposition}[definition]{Proposition}
\newtheorem{theorem}[definition]{Theorem}

\newtheorem{corollary}[definition]{Corollary}
\newtheorem{remark}[definition]{Remark}

\begin{document}

\begin{frontmatter}


\title{Reduced-Order Model Predictive Control of a Fish Schooling Model}

\author{Masaki Ogura}

\author{Naoki Wakamiya}

\affiliation{organization={Graduate School of Information Science and Technology, Osaka University},
            city={Suita},
            postcode={565-0871}, 
            state={Osaka},
            country={Japan}}

\begin{abstract}
\del{In this paper, we}\add{We} study the problem of model predictive control (MPC) for the fish schooling model \add{proposed} by Gautrais et al.~(\emph{Annales Zoologici Fennici}, 2008). \del{This typical mathematical model of fish schooling contains an attraction/alignment/repulsion law that {is not well-considered in the current results in the systems and control theory}. }The high nonlinearity of the \del{dynamics}\add{model} \del{resulting from}\add{attributed to} \del{the}\add{its attraction/alignment/repulsion} law suggests \del{using the }\add{the need to use} MPC for controlling \del{the motion of }the fish schooling\add{'s motion}. However, \del{in the case of}\add{for} large schools, \del{one can expect that }the \del{law's }hybrid nature \add{of the law} \del{makes}\add{can make} it numerically \del{difficult}\add{demanding} to perform finite-horizon optimizations in MPC. Therefore, this paper proposes \del{a reduction of}\add{reducing} the fish schooling model for numerically efficient MPC\del{. The}\add{; the} reduction is based on \del{taking}\add{using} the weighted average of the directions of individual fish in the school. \del{On this reduction, we}\add{We} analytically show \del{that}\add{how} using the normalized eigenvector centrality of the alignment-interaction network can yield a better reduction \add{by comparing reduction errors}. We confirm this finding on the weight \del{as well as}\add{and} \del{the }numerical efficiency of the MPC with the reduced-order model by numerical simulations. \add{The proposed reduction allows us to control a school with up to 500 individuals. Further, we confirm that reduction with the normalized eigenvector centrality allows us to improve the control accuracy by factor of five when compared to that using constant weights.} \end{abstract}




\end{frontmatter}


\section{Introduction}

Swarming behaviors~\cite{Perc2013,Leonard2014} \del{can be found in}\add{are exhibited by} various biological populations\del{,} including ant colonies~\cite{Burton1985}, \del{bird}\add{birds} \del{flocking}\add{flocks}~\cite{Emlen1952}, \del{and }fish \del{schooling}\add{schools}~\cite{Parrish2002}\add{, traffic flows~\cite{Feliciani2016}, and migration~\cite{Berdahl2016}}. A fundamental reason for the emergence of \del{the }swarming behaviors is \del{their advantages }that \del{are not necessarily achievable}\add{their advantages cannot be achieved} by a single individual; \add{the} examples of \del{the}\add{these} advantages include fault tolerance, parallelism, robustness, and flexibility. These advantages are specifically useful for performing complex tasks such as gathering and processing information, protection from predators, and distributed resource allocation. \del{For this reason}\add{Therefore}, engineers in various disciplines have been investigating \del{the }swarming behaviors to develop systems \del{having}\add{with} the \add{aforementioned} advantages\del{ above}. For example, \del{the }particle swarm optimization~\cite{Kennedy1995} is \del{one}\add{a} successful application in mathematical optimization inspired by \del{the }swarming behaviors. Another prominent example is \del{the }swarm robotics~\cite{Brambilla2013,Mayya2019}, \add{which is} an emerging approach \del{to}\add{for} collective robotics supported by inspirations from swarming behaviors. In the \del{context of}\add{literature on} systems and control theory, \del{the papers \mbox{\cite{Gray2018,Stella2019}} present} distributed and adaptive control laws for multiagent systems inspired by honeybees swarms \add{have been reported~\cite{Gray2018,Stella2019}}. \add{Besides the aforementioned engineering applications, swarming models offer the emerging field of social physics~\add{\cite{Jusup2022}} a major tool for modeling and analyzing social phenomena~\cite{Perc2017}.}

Schooling~\cite{Parrish2002} is \del{one of the}\add{a} typical swarming \del{behaviors taken by fishes}\add{behavior exhibited by fish}\add{.} \del{and}\add{It} provides \del{the }fish populations with several benefits, \del{including}\add{such as} protection from predators, efficient foraging, and higher rates \del{in}\add{of} finding a mate. Motivated by these direct advantages and their remarkable benefits in evolutionary processes, researchers \del{are}\add{have focused on} working toward modeling, analysis, and control of fish schools from \add{a} wide range of perspectives. In the context of fish biology, several mathematical models for fish schooling have been proposed~\cite{Huth1992,Couzin2002,Kunz2003,Viscido2004,Gautrais2008} to better model and understand \add{the} fish schooling phenomena, \del{laying the}\add{which lays the} foundations for recent data-driven analysis~\cite{Calovi2015} of fish schools. \del{In the context of river engineering, there}\add{There} exists an ongoing research trend \del{for guiding}\add{in the field of river engineering to guide} fish schools to prevent them from being trapped into harmful objects such as water diversions~\add{\cite{Poletto2015}}\del{(see, e.g.,~\mbox{\cite{Poletto2015}})}. Another fish guidance problem of practical importance can be found in fishery engineering, where researchers \del{are investigating}\add{investigate} methods \del{for improving}\add{to improve} the selectivity of net fisheries by \del{the modification of}\add{modifying} schooling behaviors with light guidance devices~\cite{Virgili2018,Bielli2020}. Finally, in the context of swarm intelligence, the authors in~\cite{Filho2009} proposed a search algorithm called \del{Fish School Search}\add{fish school search}, which is known for its high-dimensional search ability, automatic selection between exploration and exploitation, and self-adaptability.

Although several control theoretical frameworks exist for \del{the analysis and control of}\add{analyzing and controlling} swarming behaviors \del{\mbox{\cite{Olfati-Saber2006a,Su2007,Cucker2007,Li2008a,Bakshi2020,Fedele2020,Paley2007,Yi2021,Xue2020}}}\add{\cite{Olfati-Saber2006a,Li2008a,Bakshi2020,Paley2007}}, many \del{of them }are not necessarily tailored to fish schooling because the \add{aforementioned} works \del{mentioned above }do not explicitly consider the short-distance repulsion, middle-range alignment, and long-distance attraction rule \del{commonly }used \add{commonly} in the mathematical modeling of \add{the motion of} fish \add{schools}~\cite{Huth1992,Couzin2002,Kunz2003,Viscido2004,Gautrais2008}. One exception is the work by Li et al.~\cite{Li2008c}, where the authors proposed a method \del{for designing}\add{to design} decentralized controllers utilizing an attraction/alignment/repulsion law for a swarm of mobile agents to achieve a collective group behavior. However, the method assumes a decentralized controller for each individual in a swarm\del{ and, therefore,}\add{, and therefore, it} is not necessarily easy to apply for the control of fish schools. \del{We also remark that control approach via leader selection \mbox{\cite{Sharkey2017,Fitch2016,Liu2012b}} is not very realistic in our scenario of guiding fish school.}

\del{To fill}\add{We study the model predictive control (MPC)~\cite{Rossiter2017} of a fish schooling model proposed by Gautrais et al.~\cite{Gautrais2008} to fill} in the gap between the current results in the systems and control theory and the standard modeling practice of fishes in mathematical biology\del{, in this paper, we study the model predictive control (MPC) of a fish schooling model proposed by Gautrais et al.~\mbox{\cite{Gautrais2008}}}. \add{We choose this model because it is built on a standard model~\cite{Couzin2002} for swarming behaviors,  and furthermore, its treatment of the behavioral noises of individuals within the swarm is improved compared with that of the original model.} \del{Because the}\add{The} model has a discontinuity arising from \del{the }short-distance repulsion, middle-range alignment, and long-distance attraction rule\del{,}\add{; therefore} it is not necessarily realistic to \del{directly }use the model \del{for performing}\add{directly to perform} MPC. Therefore, we first propose a low-dimensional reduction of the schooling model based on the weighted average of individual directions. {In the reduced-order model, the alignment effects are encoded as another set of multiplicative weights\del{, which are} equal to the inverse of the out-degrees of the alignment-interaction network, whereas the attraction effects are reduced to a single additive input\del{-term} \add{term}.} Furthermore, our analysis of the reduction error suggests the effectiveness of setting the weight on directions to be the normalized eigenvector centrality of the alignment-interaction network of fish schooling. \del{We then}\add{Then, we} perform \del{a model predictive control}{an MPC} of fish schooling \del{by }using the reduced model for prediction. We illustrate the effectiveness of using the normalized eigenvector centrality \del{as well as}\add{and} the numerical efficiency from model reduction via numerical simulations. 

\add{The contributions of the paper can be summarized as follows. First, the paper presents reduced-order models and  their reduction errors analytically for the fish schooling model proposed by Gautrais et al.~\cite{Gautrais2008}. One reduction uses static weights to aggregate variables within the school, while the other uses dynamically changing weights determined by the normalized eigenvector centrality of an interaction network within the school. Second, we introduce several formulas for analyzing the difference between the original and reduced-order models for evaluating the reduction errors. Third, we both analytically and numerically confirm the superiority of the latter reduction using the normalized eigenvector centrality by extensive MPC simulations.}

{\del{The current}\add{This} research is motivated by \del{emerging attention}\add{the increasing attention}~\cite{Poletto2015} \del{on}\add{received by} light guidance devices for modifying fish school behaviors to prevent them from \add{being} trapped into harmful objects such as water diversions. Several experimental research \del{works}\add{studies} confirm repulsive~\cite{Richards2007,Elvidge2019} and attractive~\cite{Marchesan2005,Ford2018,Kim2019} actions of fish against light guidance devices, \del{suggesting}\add{which suggest} the devices' potential applications toward the control of fish behaviors. However, there is scarce theoretical research \del{leveraging the}\add{that leverages} qualitative findings in the aforementioned experimental research to understand how \del{the }light guidance devices should be actuated \del{to achieve}\add{for achieving} a prescribed objective on the movement of fish schools. This research aims to shed light on the problem from the perspective of systems and control theory as well as network science.}

This paper is organized as follows. \add{In Section~\ref{sec:lit}, we overview related works.} In Section~\ref{sec:model}, we \del{give}\add{present} an overview of the fish schooling model by Gautrais et al.~\cite{Gautrais2008} and formulate the tracking control problem studied in this paper. In \del{Sections~4 and~5}\add{Section~\ref{sec:reduction:static}}, we present reduced-order models for the fish schooling model based on static weights and centrality-based weights, respectively. Numerical simulations are presented in Section~\ref{sec:num}.

\del{\subsection{\del{Mathematical Preliminaries}}}

\del{For a vector~$x\in \mathbb R^3$ and a constant $\theta \geq 0$, we define the operator }
\del{\begin{equation*}
    \del{R_{x, \theta} \colon \mathbb R^{3}\backslash \{0\}
    \to 
    \mathbb R^{3}\backslash \{0\}}
\end{equation*}}
\del{that rotates a given nonzero vector toward the direction of~$x$ by angle~$\theta\geq 0$. If $x=0$, then we regard $R_{x, \theta}$ as the identity operator. We remark that the rotation caused by $R_{x, \theta}$ is not supposed to terminate at the direction of~$x$; therefore, $R_{x, 2\pi}$ coincides with the identity operator. Then, for vectors $x, y \in \mathbb R^3$, we define }
\del{\begin{equation*}
    \del{\angle(x, y)
    =}
    \begin{cases}
    \del{0,}&\del{\mbox{if $x=0$ or $y=0$}, }
    \\
    \del{\mbox{angle between $x$ and~$y$},}&\del{\mbox{otherwise}.}
    \end{cases}
\end{equation*}}
\del{We also define the normalization operator}
\del{\begin{equation*}
    \del{\phi\colon \mathbb R^3\backslash \{0\} \to \mathbb R^3\backslash \{0\} \colon x\mapsto \frac{x}{\lVert x \rVert}, }
\end{equation*}}
\del{where $\norm{\cdot}$ denotes the Euclidean norm.}
\del{For a fixed $x\neq 0$, a simple calculation gives us a second-order estimate} 
\del{\begin{equation*}
    \del{\Norm{\phi(x+y) - \left(
    \phi(x) + \frac{1}{\norm{x}} y - \frac{y^\top x}{\norm{x}^{3}}x 
    \right)}=o(\norm{y}^2)}
\end{equation*}}
\del{as $y \to 0$. The following lemma \del{on this operator}{makes this estimate precise and} plays an important role in this paper.} 

\del{\begin{lemma}
\del{For all $x, y \in \mathbb R^3$ such that $x\neq 0$ and~$x+y\neq 0$, define }
\del{\begin{equation*}
    \del{\delta\phi(x, y) = \phi(x+y) - \left(
    \phi(x) + \frac{1}{\norm{x}} y - \frac{y^\top x}{\norm{x}^{3}}x 
    \right).}
\end{equation*}}
\del{Then, there exists a scalar $C \geq 0$ such that }
\begin{equation*}
\del{\norm{\delta \phi(x, y)} < \frac{C}{\norm x^2}\norm{y}^2 }
\end{equation*}
\del{for all $x\neq 0$ and~$y \in \mathbb R^3 \backslash \{-x\}$.} 
\end{lemma}}

\del{\begin{proof}
\del{See \ref{app:}.}
\end{proof}}

\del{For a set of vectors $v_1, \dotsc, v_N$ with the same dimension and scalars $\alpha_1, \dotsc, \alpha_N$, we often use the shorthand notation }
\del{\begin{equation*}
    \del{\langle v\rangle_\alpha = \sum_{i=1}^N \alpha_i v_i.} 
\end{equation*}}
\del{In the special case where $\alpha_1 = \cdots = \alpha_N = c$ for a constant~$c$, we write}
\del{\begin{equation*}
    \del{\langle v\rangle_c = c \sum_{i=1}^N v_i.}
\end{equation*}}

\section{\add{Related Work}}\label{sec:lit}

\add{Although MPC approach \cite{Rossiter2017} offers many advanced features over classical control techniques, the approach suffers from an important and fundamental challenge of computational complexity in optimization. This complexity is caused by MPC's nature of  receding horizon, which requires that  optimization problems be solved in real time. One approach to solve this problem is model reduction and it can reduce the complexity of the optimization problems. Hovland et al.~\cite{Hovland2008} presented a framework for the constrained optimal real-time MPC of linear systems combined with model reduction in the output feedback implementation. A reduced-order model is derived from a goal-oriented formulation that enables efficient control. Lohning et al.~\cite{Lohning2014} presented a novel MPC scheme for linear systems based on reduced-order models for prediction and combined them with an error-bounding system. Explicit time- and input-dependent bounds on the model-order reduction error are employed for realizing constraint satisfaction, recursive feasibility, and asymptotic stability. Recently, Lorenzetti et al.~\cite{Lorenzetti2022} proposed a reduced-order MPC scheme to solve robust output feedback constrained optimal control problems for high-dimensional linear systems. They introduced a projected reduced-order model, which improves computational efficiency and provides robust constraint satisfaction and stability guarantees.}

\add{Some results on the MPC of nonlinear systems using model reduction have been reported in the literature. For example, Wiese et al.~\cite{Wiese2015} presented an MPC method for gas turbines. They specifically developed a lower-order internal model from a physics-based higher-order model using rigorous time-scale separation arguments that can be extended to various gas turbine systems. Zhang and Liu~\cite{Zhang2019c} considered the problem of the economic model predictive control of wastewater treatment plants based on model reduction. Their reduction is based on a technique called reduced-order trajectory segment linearization. The authors showed via numerical simulations that, while the proposed methods lead to improved computational efficiency, they do not involve reduced control performance.}

\add{Network science has introduced various types of centrality measures to determine the relative importance of nodes in a network under respective circumstances~\cite{Boldi2014}. Given this context, the interplay between control and centrality has been  actively investigated. Liu et al.~\cite{Liu2012b}
introduced the concept of control centrality to quantify the ability of a single node to control the entire network. Inspired by the  relationship between control centrality and the hierarchical structure in networks, the authors  designed efficient attack strategies against the controllability of malicious networks.  Fitch et al.~\cite{Fitch2016} showed that  the tracking performance of any leader set within a multiagent system can be quantified by a novel centrality measure called joint centrality. For both single and multiple leaders, the authors have analytically proven the effectiveness of the centrality measure for leader selection.}

\section{Fish Schooling Model}\label{sec:model}

In this section, we \del{first give}\add{provide} an overview of the fish schooling model presented by Gautrais et al.~\cite{Gautrais2008} in Section~\ref{subsec:auto}. \del{We then}\add{Then, we} present a model of a \emph{controlled} fish schooling model in Section~\ref{subsec:dontrol}. Finally, in Section~\ref{subsec:track}, we state the tracking problem \del{that is }studied in this paper. 

\subsection{Autonomous Model}\label{subsec:auto}

\add{The mathematical models of swarming behavior can be divided into two categories: the Eulerian approach~\cite{Bayen2006} and the Lagrangian approach~\cite{Kumar2021a,Kumar2021}. The Eulerian approach considers the swarm to be a continuum described by its density, and its time evolution is described by a partial differential equation. The Lagrangian approach considers the state of each individual (position, instantaneous velocity, and instantaneous acceleration) and its relationship to other individuals in the swarm. It is an individual-based approach, and therefore, velocity and acceleration can be influenced by the spatial coordinates of the individuals. The time evolution of the state is described by ordinary or stochastic differential equations.}

The fish schooling model by Gautrais et al.~\cite{Gautrais2008} \add{employs the Lagrangian approach and} consists of a set of fish considered as point masses in the three dimensional Euclidean space~$\mathbb R^3$. We denote \del{by $N$ }the number of fish \add{by $N$}. \add{Within the model, fish agents move in the three dimensional space~$\mathbb R^3$ in discrete time steps.} For each \add{index}~$i\in \{1, \dotsc, N\}$ \add{and time~$k = 0, 1, 2, \dots$}, we let $x_i(k) \in \mathbb R^3$ denote the position of the $i$th fish, and \add{we} let $V_i(k) \in \mathbb R^3$ denote the unit vector \del{representing}\add{that represents} the direction of the $i$th fish. In the schooling model, the positions of fish are dynamically updated by\del{ the difference equation}
\begin{equation}\label{eq:_referred_xv}
    x_i(k+1) = x_i(k) +  \tau v 
    V_i(k),\ i=1, \dotsc, N, 
\end{equation}
where $v > 0$ denotes the speed per unit time of fish and~$\tau > 0$ represents the length of the interval for discretizing the motion of fish. 

In the schooling model, each fish recognizes its neighbor fish to dynamically update its direction. \del{Specifically, it is supposed that the}\add{The} $i$th fish \del{has}\add{is assumed to have} the following mutually disjoint sets for recognizing other fishes\del{;}\add{:} the region of repulsion $Z_{r, i}(k)\subset \mathbb R^3$, \del{the }region of orientation $Z_{o, i}(k)\subset \mathbb R^3$, and \del{the }region of attraction $Z_{a, i}(k)\subset \mathbb R^3$. \del{The precise definition of these regions in~\mbox{\cite{Gautrais2008}} is stated later in this subsection. }Using these sets, we define the set of neighbors of the $i$th fish by 
\begin{equation*}
    \begin{aligned}
    \mathcal N_{r, i}(k) & = 
    \{j \in \{1, \dotsc, N\}\backslash \{i\} \mid x_j(k) \in Z_{r, i}(k) \},\notag 
    \\
    \mathcal N_{o, i}(k) & = 
    \{j \in \{1, \dotsc, N\}\backslash \{i\} \mid x_j(k) \in Z_{o, i}(k) \},\\
    \mathcal N_{a, i}(k) & = 
    \{j \in \{1, \dotsc, N\}\backslash \{i\} \mid x_j(k) \in Z_{a, i}(k) \}. 
    \end{aligned}
\end{equation*}
We then define \del{the following vectors dependent on each fish~$i$:}\add{the forces of repulsion, orientation, and attraction that can act on the $i$th fish at time~$k$ by} 
\begin{align}
    \del{R}{E}_i(k) &= -\sum_{j\in \mathcal N_{r, i}(k)}\phi(x_j(k)-x_i(k)),\label{eq:def:R_i}  
    \\
    O_i(k) &= \sum_{j\in \mathcal N_{o, i}(k)}V_j(k),\label{eq:def:O_i} 
    \\
    A_i(k) &= \sum_{j\in \mathcal N_{a, i}(k)}\phi(x_j(k)-x_i(k)), \label{eq:def:A_i}
\end{align}
\del{which represent the forces of repulsion, orientation, and attraction that can act on the $i$th fish at time~$k$}\add{where
\begin{equation}\label{eq:def:phi}
    \phi\colon \mathbb R^3\backslash \{0\} \to \mathbb R^3\backslash \{0\} \colon x\mapsto \frac{x}{\lVert x \rVert}\del{,} 
\end{equation}
is the normalization operator with respect to the Euclidean norm~$\norm{\cdot}$}. These forces are used to determine the \emph{desired direction} of the $i$th fish at time~$k$ as
\del{\begin{equation*}
    \del{D_i(k) =}
    \begin{cases}
    \del{O_i(k) + \eta A_i(k),}&\del{\mbox{if $\mathcal N_{r, i}(k) = \emptyset$,}}
    \\
    \del{R}\del{{E}_i(k),}&\del{\mbox{otherwise,}}
    \end{cases}
\end{equation*}}
\begin{equation}\label{eq:D_i:gautrais}
    \add{D_i(k) =
    \begin{cases}
    {E}_i(k),&\mbox{if $\mathcal N_{r, i}(k)\neq \emptyset$,}
    \\
    O_i(k) + \eta A_i(k),&\mbox{if $\mathcal N_{r, i}(k) =  \emptyset$ and~$\mathcal N_{o, i}(k) \cup \mathcal N_{a, i}(k) \neq \emptyset$,}
    \\
    V_i(k),&\mbox{otherwise,}
    \end{cases}}
\end{equation}
where $\eta>0$ is a constant. \add{Equation \eqref{eq:D_i:gautrais} implies that the repulsion force~$E_i$ prevails if there is at least one neighbor in the region of repulsion. The $i$th fish tends to align its direction with the neighbors in the region of orientation or to stay close to the neighbors in the region of attraction when there is no neighbor in the region of repulsion but some neighbor exist in the regions of orientation or attraction, If there are no neighbors within any of the regions, the $i$th fish desires to maintain its direction.}

\del{Then, the}\add{The} authors in~\cite{Gautrais2008} introduce the following update law for the direction of fish. In the update law, the fish first rotates toward the desired direction with a turning rate~$\theta>0$\add{,} and\del{,} then, \add{it} performs a random rotation. \add{Let us first introduce some notations to formulate the update law. For a vector~$x\in \mathbb R^3$ and a constant $\theta \geq 0$, we define the operator 
\begin{equation*}
    R_{x, \theta} \colon \mathbb R^{3}\backslash \{0\}
    \to 
    \mathbb R^{3}\backslash \{0\}
\end{equation*}
that rotates a given nonzero vector toward the direction of~$x$ by angle~$\theta\geq 0$. The rotation caused by $R_{x, \theta}$ is not supposed to terminate at the direction of~$x$; therefore, $R_{x, 2\pi}$ coincides with the identity operator. We further introduce the following notation to denote the angle between vectors:}
\begin{equation*}
    \add{\angle(x, y)
    =}
    \begin{cases}
    \add{0,}&\mbox{\add{if $x=0$ or $y=0$}}, 
    \\
    \add{\mbox{angle between $x$ and~$y$},}&\add{\mbox{otherwise}. }
    \end{cases}
\end{equation*}
\del{This}\add{Now, the} update law \add{for the direction of fish} can be \del{formulated}\add{stated} as 
\begin{equation}\label{eq:rotation}
    V_i(k+1) = R_{\del{\epsilon_i(k), }\chi_i(k){, \epsilon_i(k)}} R_{\del{\phi_i(k),  }D_i(k){, \phi_i(k)}}V_i(k), 
\end{equation}
where
\begin{equation*}
    \phi_i(k) = \min\bigl(\angle\left(V_i(k), D_i(k)\right), \tau\theta\bigr)
\end{equation*}
represents the rotation angle toward the desired direction with turning rate~$\theta$, $\epsilon_i(k)$ represents the angle of the random rotation, and~$\chi_i(k)$ \del{is}\add{represents} the random vector drawn from the uniform distribution over the sphere. \add{The random rotation is introduced to model internal errors and external disturbances arising from, for example, the lack of precision in the perception processes.}

\begin{figure}[tb]
\centering
\add{\includegraphics[width=.45\linewidth]{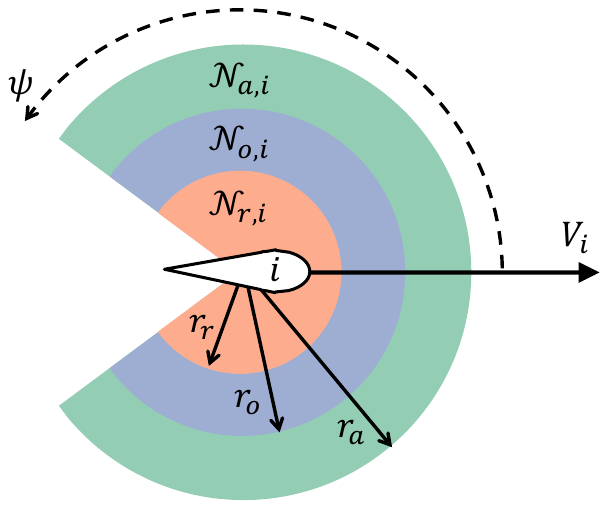}}
\caption{\add{An individual is assumed to be able to evaluate the relative positions of other individuals within the sensing region that has radius $r_a$ and cone angle~$2\psi$. The sensing region is divided into the region of repulsion, orientation, and attraction.%
}}
\label{fig:fish}
\end{figure}

Finally, we \del{below }describe the definitions of the regions of repulsion, orientation, and attraction used in~\cite{Gautrais2008}. We first introduce a constant $\psi>0$, which represents the viewing angle of the fishes. Then, we define the cone 
$C_i(k) = 
\{
z \in \mathbb R^3  \mid \angle(V_i(k), z-x_i(k))\leq \psi
\}$. 
Under this notation, the authors in~\cite{Gautrais2008} adopt the \del{following }definition \add{(see Fig.~\ref{fig:fish} for an illustration)}\del{:} 
\begin{equation*}
\begin{aligned}
    Z_{r, i}(k)  & = C_i(k) \cap \{
    z \in \mathbb R^3 \mid 0
    < \lVert z-x_i(k)\rVert \leq r_r
    \}, 
    \\
    Z_{o, i}(k) & = C_i(k) \cap \{
    z \in \mathbb R^3 \mid r_r
    < \lVert z-x_i(k)\rVert \leq r_o
    \}, 
    \\
    Z_{a, i}(k) & = C_i(k)  \cap \{
    z \in \mathbb R^3 \mid r_o
    < \lVert z-x_i(k)\rVert \leq r_a
    \}, 
\end{aligned}
\end{equation*}
where the positive constants $r_r$, $r_o$, and~$r_a$ are called the radius of repulsion, orientation, and attraction, respectively, and \add{they} are supposed to satisfy $0<r_r<r_o<r_a$. \add{The angle~$\psi$ is assumed to be strictly less than~$\pi$ to model a blind zone in the visual field of fish.}

\label{interpretation of the model}\add{The current model extends the well-known Boid model~\cite{Reynolds1987} by focusing on fish schooling. Further, the current model extends a previous model~\cite{Couzin2002} for swarming behaviors by improving the treatment of noise. In the current model, a noise is designed to  always affect the direction of fish, which is not the case in the model presented in~\cite{Couzin2002}.}

\subsection{Controlled Model}\label{subsec:dontrol}

The fish schooling model \add{proposed} by Gautrais et al.~\cite{Gautrais2008} does not consider the possibility of external stimulus to the fish school. In this paper, we adopt a simple modification of the model to incorporate external stimulus. For simplicity, we assume that the stimulus attracts the $i$th fish to the direction of a unit vector~$u_i\in\mathbb R^3$. \del{We also}\add{Further, we} assume that each fish has its own sensitivity parameter~$\xi_i \geq 0$, which can account for the heterogeneity within the swarm of fishes. We then propose that the original equation~\eqref{eq:D_i:gautrais} for the desired direction be modified as\del{ follows:} 
\begin{equation}\label{eq:D_i:gautrais+control}
    D_i(k) =
    \begin{cases}
    O_i(k) + \eta A_i(k)+ \xi_i u_i(k),&\mbox{if $\mathcal N_{r, i}(k)= \emptyset$,}
    \\
    \del{R}{E}_i(k),&\mbox{otherwise.}
    \end{cases}
\end{equation}
\add{We introduce the notation~$u_i$ for the generality of the theory developed in this paper; although the notation~$u_i$ suggests the possibility of applying stimulus~$u_i$ designed independently for each individual, realizing such a stimulus is not realistic. In reality, the degree of freedom in designing the set of stimulus~$\{u_i\}_{i=1}^N$ can be severely limited.}\label{whythisassumptionisintroduced}

\subsection{Tracking Problem}\label{subsec:track}

\del{For}\add{We consider a tracking problem for} the controlled schooling model introduced in the last subsection\del{, we consider a tracking problem in this paper}. Assume that a reference set~$\mathcal R \subset \mathbb R^3$ is given. Our objective is to design the external stimulus $u_i$ ($i=1, \dotsc, N$) appearing in \eqref{eq:D_i:gautrais+control} to guide the fish school and make the mass center 
\begin{equation}\label{eq:isreferred_def:center}
    c = \frac{1}{N}\sum_{i=1}^N x_i
\end{equation}
of the fish school lie on the set~$\mathcal R$. The closeness is measured by the distance 
\begin{equation}\label{eq:defe}
e(k) = \dist (c(k), \mathcal R ), 
\end{equation}
where $\dist(\cdot, \cdot)$ denotes the distance between a vector and a subset of~$\mathbb R^3$ and is defined by $\dist(y, \mathcal Z) = \inf\{\norm{y-z} \mid z \in \mathcal Z \}$. 
 
The following \del{situation}\add{control problem} is considered in this paper. 
We assume that, at every $T$ discrete steps (i.e., at every $T \tau$ units of time in the continuous time), we can observe the fish school to obtain \del{the }positions and directions of all fish. \del{We also}\add{Further, we} assume that we can update the value of the external stimulus $u_i$ only at the time of these observations; i.e., for all $i\in \{1, \dotsc, N\}$ and \del{$k\geq 0$}\add{$\ell\geq 0$}, we impose the constraint
\begin{equation}\label{eq:isreferred_periodicUpdate}
    \del{u_i(k) = u_i(\lfloor k/T\rfloor T),}
    \add{u_i(\ell T) = u_i(\ell T+1)= \cdots = u_i((\ell+1) T -1).}
\end{equation}
\del{where $\lfloor \cdot \rfloor$ denotes the floor function. }We \del{suppose that we }use the observation $\{x_i(\ell T), V_i(\ell T)\}_{i=1, \dotsc, N}$ at the $\ell$th step for determining the stimulus within the subsequent sampling interval (i.e., at times $k=(\ell+1)T, (\ell+1)T+1, \dotsc, (\ell+2)T-1$). Therefore, it is required that the stimulus computation finishes within $T\tau$ units of time in the continuous time. 

The high nonlinearity of the schooling model suggests that the framework of the MPC is \del{most }appropriate \del{for determining}\add{to determine} the stimulus. Therefore, let us consider the \del{situation}\add{scenario} where we determine the stimulus by iteratively solving finite-horizon optimization problems to minimize the cost function 
\begin{equation*}
    \sum_{k=(\ell+1)T\del{+1}}^{(\ell+1)T+T_h}{e(k)}, 
\end{equation*}
where\del{$\hat c(k)$ represents the mass center of the fish school determined by a prediction model and} $T_h > 0$ denotes the length of the optimization horizon. 

One obvious option for prediction is to use the original schooling model. In this prediction model, one can use the original schooling dynamics as is except \add{with} the rotation dynamics~\eqref{eq:rotation} replaced by
\begin{equation}\label{eq:rotation:withoutrandom}
    V_i(k+1) =  R_{\del{\phi_i(k),  }D_i(k){, \phi_i(k)}}V_i(k), 
\end{equation}
where the random rotation operator~$R_{\del{\epsilon_i(k), }\chi_i(k){, \epsilon_i(k)}}$ is set to the identity operator. \del{Let us denote this}{This} prediction model{, denoted} by~$\Sigma_{\orig}${, results in the following optimization problem to be solved during the continuous-time interval $[\ell T \tau, (\ell+1) T \tau)$:
\begin{equation}\label{eq:clarifiedoptimizationproblem}
    \begin{aligned}
        \minimize_{\{u(k)\}_{k=(\ell +1)T}^{(\ell +1)T+T_h-1}} &\del{\eqref{eq:costFunction}}\add{ \sum_{k=(\ell+1)T\del{+1}}^{(\ell+1)T+T_h}{e(k)}}
        \\
        \subjectto &
        \mbox{\cref{eq:_referred_xv,eq:def:R_i,eq:def:O_i,eq:def:A_i,eq:D_i:gautrais+control,eq:isreferred_periodicUpdate,eq:rotation:withoutrandom,eq:isreferred_def:center,eq:defe}}
    \end{aligned}
\end{equation}
for determining the control input  $\{u(k)\}_{k=(\ell+1)T}^{k=(\ell+2)T-1}$ to be used in the next (continuous-time) interval $[(\ell+1) T \tau, (\ell+2) T \tau)$.} \add{A schematic of the control algorithm is illustrated in Fig.~\ref{fig:MPCorig}.} Although {the optimization problem \eqref{eq:clarifiedoptimizationproblem} based on} \del{the prediction model }$\Sigma_{\orig}$ could allow us to realize precise tracking \del{of}\add{by} the fish school, a natural concern in this context is the computational cost for performing predictions, \del{as}\add{because} the number of inter-fish interactions in the model grows quadratically with respect to the size~$N$ of the school. This potential numerical difficulty is illustrated in Section~\ref{sec:num}. \del{Therefore}\add{Thus}, it is desirable to construct reduced-order prediction models for performing numerically efficient MPC. 

\begin{figure}[tb]
    \centering
    \add{\includegraphics[width=.725\linewidth]{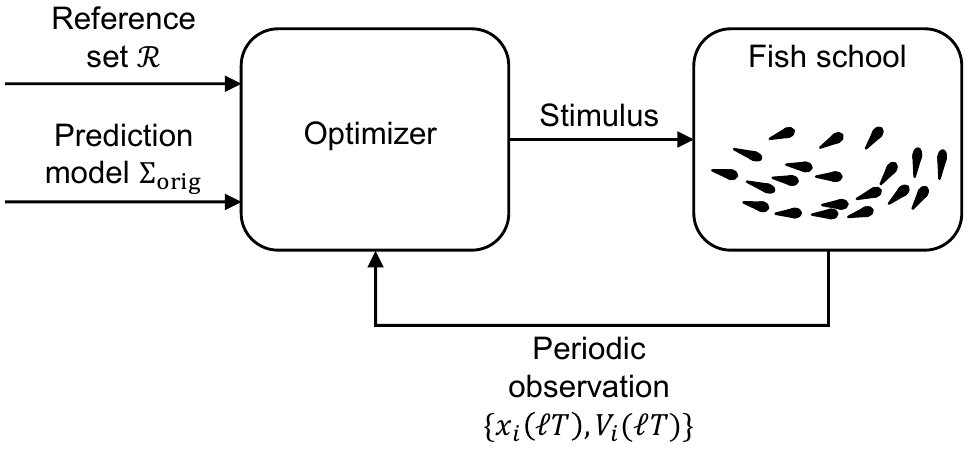}}
    \caption{\add{Schematic of MPC with the original schooling model~$\Sigma_{\orig}$. The positions and directions of each individual are measured periodically. Using the measurement, reference set~$\mathcal R$, and model~$\Sigma_{\orig}$, the optimizer periodically updates the stimulus applied to individuals.}}
    \label{fig:MPCorig}
\end{figure}

\section{\del{Static weight}\add{Reduced-order MPC}}\label{sec:reduction:static}

\add{We present reduced-order models for the fish schooling model based on static weights and centrality-based weights, respectively. Based on the reductions, we then present reduced-order MPC algorithms for solving the tracking problem introduced in the last section. Throughout this section, we use the following notation. For a set of vectors $v_1, \dotsc, v_N$ with the same dimension and scalars $\alpha_1, \dotsc, \alpha_N$, we write}
\begin{equation*}\label{notation}
    \add{\langle v\rangle_\alpha = \sum_{i=1}^N \alpha_i v_i.} 
\end{equation*}
\add{In the special case where $\alpha_1 = \cdots = \alpha_N = p$ for a constant~$p$, we write}
\begin{equation}\label{cnotation}
    \add{\langle v\rangle_{p} = p \sum_{i=1}^N v_i.}
\end{equation}

\subsection{\add{Static Weight}}

\del{In this section, we}\add{We} propose a static weight-reduction of the model~$\Sigma_{\orig}$ {for predicting the trajectory of the mass center~$c(k)$. \add{We begin by outlining our overall idea.} Under \del{the} notation~\eqref{cnotation}, we have
\begin{equation}\label{eq:cupdate}
    c(k+1) = c(k) + \tau v \langle V(k)\rangle_{1/N}.
\end{equation}
\add{Therefore, if we could construct an update low for $\langle V\rangle_{1/N}$ in terms of~$c$ and~$\langle V\rangle_{1/N}$, then the update law in conjunction with \eqref{eq:cupdate} serves as a reduced-order model for predicting the movement of the mass center. To proceed, let us tentatively consider the scenario in which}\del{If} the rotation speed~$\theta$ is relatively high\del{,}\add{. In this situation,} one can expect that the fish in the school can easily rotate to align with their desired directions. Indeed, for \add{a} sufficiently large $\theta$, equation~\eqref{eq:rotation:withoutrandom} implies $V_i(k+1) = \phi(D_i(k))$, which further implies \del{$\langle V(k+1) \rangle_{1/N} = \langle \phi(D(k)) \rangle_{1/N}$.}
\begin{equation*}
    \langle V(k+1) \rangle_{1/N} = \langle \phi(D(k)) \rangle_{1/N}.
\end{equation*}
\add{Therefore, approximation of~$\langle \phi(D(k)) \rangle_{1/N}$ in terms of~$c$ and~$\langle V \rangle_{1/N}$ allows us to construct a reduced-order model for prediction of the mass center.}

Motivated by this observation, \del{in this section }we propose an approximation formula for computing the quantity~$\langle \phi(D(k)) \rangle_{1/N}$}. \del{We specifically propose a reduced-order model for prediction based on the weighted sum~$\langle V(k) \rangle_\alpha = \sum_{i=1}^N \alpha_i V_i(k)$
of the direction vectors, where}{For the sake of generality of the theoretical development, \del{throughout this section, }we allow the weight~$1/N$ to be arbitrary} nonnegative numbers {$\alpha_1$, \dots, $\alpha_N$} satisfying
\begin{equation}\label{eq:sumtoone}
    \sum_{i=1}^N \alpha_i = 1. 
\end{equation}
\label{rolealpha}

\del{
\addtocounter{subsection}{-1}
\subsection{
\del{Reduction Error Analysis}
}
}

\subsubsection{\add{Reduction Error Analysis}}

We \del{begin by preparing}\add{present} some notations. Let $i, j\in \{1,\dotsc, N\}$ and~$k\geq 0$ be arbitrary. \del{Define the integer}\add{For the $i$th fish, let}
\begin{equation*}
    n_i(k) = \abs{\mathcal N_{o, i}(k)}\del{.}
\end{equation*}
\add{denote the number of its neighbors within the region of orientation.}
Define $\omega_{ij}(k) \in \{0, 1\}$ by 
\begin{equation*}
    \omega_{ij}(k) = 
    \begin{cases}
    1,&\mbox{if $j \in \mathcal N_{o, i}(k)$}, 
    \\
    0,&\mbox{otherwise.}
    \end{cases}
\end{equation*}
\del{Also}\add{In addition}, let $\theta_{ij}(k) = \angle(V_i(k), V_j (k))$ \add{denote the angle between the directions of the $i$th and~$j$th fish}. Then, for a vector~$y = [y_1,\, \dotsc , y_N]^\top \in \mathbb R^N$, we define \del{the}\add{a} quadratic \del{form}\add{function of the angles as} 
\begin{equation*}
\begin{multlined}[.85\linewidth]
    Q_{i, y}(k) = 
    \frac{1}{n_i(k)^2}\sum_{j=1}^N \sum_{j'=1}^N \frac{\abs{\omega_{ij}(k)-y_j}\abs{\omega_{ij'}(k)-y_j'}}{2}\theta_{jj'}(k)^2. 
    \end{multlined}
\end{equation*}

\del{We also}\add{Further, we} introduce constants \del{playing}\add{that play} an important role in our analysis. 
First, define  
\begin{equation*}
    \tilde \pi_i(k) = 1-\pi_i(k) \in [0, 1], 
\end{equation*}
where $\pi_i(k) = \norm{
O_i(k)
}/n_i(k) \in [0, 1]$ 
\del{is}\add{represents} the polarization~\cite{Vicsek1995,Couzin2002,Gautrais2008} of the fish in the region of orientation $Z_{o, i}(k)$. We use the constant~$\tilde \pi_i(k)$ to quantify how aligned the directions of the fish in the region of orientation $Z_{o, i}(k)$ \add{are}; the smaller $\tilde \pi_i(k)$, the more aligned \add{are} the direction\add{s} of the fish in the region. Similarly, we define 
\begin{equation*}
    \tilde \pi_\alpha(k) = 1-\norm{\langle V(k)\rangle_\alpha}\in [0, 1]\del{.}\add{,} 
\end{equation*}
\add{which quantifies how aligned the directions of all the fish \add{are} when weighted with $\alpha$.}
We \del{also }define the ratio 
\begin{equation*}
\rho_i(k) = \frac{\norm{\eta A_i(k)+\xi_i u_i(k)}}{\norm{O_i(k)}} \geq 0. 
\end{equation*}
This quantity measures how dominant the force of orientation is in the desired direction~\eqref{eq:D_i:gautrais+control} when repulsion is \del{not present}\add{absent}. For example, if $\rho_i(k) = 0$, \del{then }both $A_i(k)$ and~$u_i(k)$ are zero\del{ and, therefore,}\add{, and therefore,} the movement of the $i$th fish is completely determined by the orientation term~$O_i(k)$. \del{Finally, we define}
\del{\begin{equation*}
\begin{multlined}
    \del{C_\phi = \inf\{
    C\geq 0\mid \mbox{\eqref{old:eq:f(x,y)bound} holds for all $(x, y)\in \mathbb R^3\times \mathbb R^3$}}
    \\\del{\mbox{such that $x\neq 0$ and~$x+y\neq 0$} 
    \},}
\end{multlined}
\end{equation*}}
\del{whose existence is guaranteed by Lemma~\ref{old:lem:2}.} 

\del{We are now ready to state the \add{first} main result of this section. }\add{Our objective is to propose an approximation formula for $\langle \phi(D(k))\rangle_\alpha$ weighted by $\alpha$. {Let us} consider the approximation 
\begin{equation}\label{eq:approx}
\langle \phi(D(k)) \rangle_\alpha \approx \norm{\langle V(k)\rangle_\alpha}\,\phi\bigl(\langle V(k)\rangle_\alpha + A(k) + u(k)\bigr)\del{,}  
\end{equation}
\add{with vectors} 
\begin{equation}\label{eq:def}
    \add{A(k) =  \sum_{i=1}^N \frac{\alpha_i}{n_i(k)} \eta A_i(k),\ 
    u(k) = \sum_{i=1}^N \frac{\alpha_i }{n_i(k)}\xi_i u_i(k)} 
\end{equation}
\add{which represent the aggregated force of attraction and external stimulus, respectively. Now, the following theorem provides a bound for the approximation~\eqref{eq:approx}.}
Within the theorem, }\del{The}\add{the} dependence of relevant variables on the time index $k$ is omitted for brevity\del{in the theorem}. 

\begin{theorem}\label{thm:static}
Assume that 
\begin{equation}\label{eq:empty}
 \mathcal N_{r, i} = \emptyset   
\end{equation}
and 
\begin{equation}
{n_i > 0,} \quad 
O_i \neq 0    \label{eq:OIneq0}
\end{equation}
for all $i =1, \dotsc, N$. {Let $\alpha_1$, \dots, $\alpha_N$ be nonnegative constants satisfying \eqref{eq:sumtoone}.} \del{Define}\add{If $\langle V \rangle_\alpha \neq 0$, then the norm of} the \add{residual} vector 
\begin{equation}\label{eq:eps}
\epsilon_\alpha = \langle \phi(D) \rangle_\alpha - \norm{\langle V\rangle_\alpha}\,\phi\bigl(\langle V\rangle_\alpha + A + u\bigr)\del{,}
\end{equation}
\del{where }\del{\begin{equation*}
    \del{A =  \sum_{i=1}^N \frac{\alpha_i }{n_i} \eta A_i,\ 
    u = \sum_{i=1}^N \frac{\alpha_i }{n_i}\xi_i u_i.}
\end{equation*}}
\del{If $\langle V \rangle_\alpha \neq 0$, then}\add{satisfies}
\begin{equation}\label{eq:error:alpha}
\begin{multlined}[.85\linewidth]
    \norm{\epsilon_\alpha} 
    \leq 
\del{C_\phi}\add{2}
    \frac{\norm{A+u}^2}{\norm{\langle V\rangle_\alpha}}    +
    \sum_{i=1}^N
    \alpha_i\left(\tilde \pi_i + \sqrt{Q_{i, n_i \alpha}}\right)
    \\+
    \add{2}\sum_{i=1}^N
    \alpha_i\rho_i\left(\del{2}\tilde \pi_\alpha + \del{C_\phi} \rho_i +\del{4}\add{2} \tilde \pi_i + \del{2} \sqrt{Q_{i, n_i \alpha}}\right).
    \end{multlined}
\end{equation}
\end{theorem}

\begin{proof}
\add{See \ref{sec:prthmstatic}.}
\end{proof}

\begin{remark}
\del{The condition}\add{Condition} \eqref{eq:empty} in Theorem~\ref{thm:static} is motivated and supported by our empirical finding that the repulsion term $R_i(k)$ defined in \eqref{eq:def:R_i} does not become active frequently. \del{Likewise,}\add{Similarly,} the second condition~\eqref{eq:OIneq0} is not very restrictive when a given group of fish form one school. 
\end{remark}

\del{\del{Theorem~\ref{thm:static} gives an upper bound on the error arising from the approximation }
\del{\begin{equation*}
\del{\langle \phi(D(k)) \rangle_\alpha \approx \norm{\langle V(k)\rangle_\alpha}\,\phi\bigl(\langle V(k)\rangle_\alpha + A(k) + u(k)\bigr),}
\end{equation*}}
\del{which motivates us to construct our reduced-order model for $\Sigma_{\orig}$ in Section~\ref{subsec:staticPrediction}.}}

Before presenting the proof of the theorem, we \del{here }briefly provide an interpretation of the bound. The bound suggests the following {two} \del{situations}\add{scenarios} in which the approximation~\eqref{eq:approx} is precise: 
\begin{itemize}
    \item \del{the}\add{The} force of attraction $\eta A_i$ \del{is}{and the stimulus $\xi_i u_i$ are} small\del{;} 
    \item \del{the}\add{The} direction of individuals are aligned\del{.} 
\end{itemize}
{Specifically, the following corollary illustrates which terms in the bound are critical in these two situations.
\begin{corollary}\label{cor:static}
Let $\alpha_1$, \dots, $\alpha_N$ be nonnegative constants satisfying \eqref{eq:sumtoone}. Assume that inequalities~\eqref{eq:empty} and~\eqref{eq:OIneq0} hold. If $A_i=u_i=0$ for all $i\in \{1, \dotsc, N\}$, then 
\begin{equation}\label{eq:stat:au0}
    \norm{\epsilon_\alpha} 
    \leq 
    \sum_{i=1}^N
    \alpha_i\left(\tilde \pi_i + \sqrt{Q_{i, n_i \alpha}}\right). 
\end{equation}
On the other hand, if $\theta_{ij} = 0$ for all $i, j\in\{1, \dotsc, N\}$, then 
\begin{equation}\label{eq:stat:theta0ineq}
    \norm{\epsilon_\alpha} 
    \leq 
    \del{C_\phi}\add{2}\left(
    \norm{A+u}^2 + \sum_{i=1}^N
    \alpha_i\norm{\eta A_i+\xi_i u_i}^2
    \right).
\end{equation}
\end{corollary}}

\del{
    \subsection{
        \del{Proof of Theorem~\mbox{\ref{thm:static}}}}
    \addtocounter{subsection}{-1}
}

\del{We present the proof of Theorem~\ref{thm:static}. For ease of the presentation of the proof, we omit writing the time variable~$k$ if no confusion can be expected. We start by proving the following technical lemma.}
\del{\begin{lemma}
\label{old:lem:deltaOi}
\del{For all $i\in \{1, \dotsc, N\}$, the vector 
    $\delta_\alpha O_i = O_i - n_i \langle V\rangle _\alpha$
satisfies}
\begin{equation*}
    \del{\norm{\delta_\alpha O_i} \leq n_i \sqrt{Q_{i, n_i \alpha}}.}
\end{equation*}
\end{lemma}}

\del{\begin{proof}
\del{For all $i, j \in \{1, \dotsc, N\}$, define $\sigma_{ij} = 1 - V_i^\top V_j$. The definition \eqref{eq:def:O_i} of the vector~$O_i$ shows $\delta_\alpha O_i = \sum_{j=1}^N(\omega_{ij}-n_i\alpha_j)V_j$. Therefore, we can show }
\begin{equation*}
    \del{\norm{\delta_\alpha O_i}^2 
=
    -\sum_{j=1}^N \sum_{j'=1}^N 
     (\omega_{ij}-n_i\alpha_j) (\omega_{ij'}-n_i\alpha_{j'}) \sigma_{jj'}, }
\end{equation*}
\del{where we used equality $\sum_{j=1}^N (\omega_{ij} - n_i \alpha_j) = 0$. Therefore, equation~\eqref{old:eq:thisequation} and the trivial inequality $\sigma_{ij} \leq \theta_{ij}^2/2$ prove $\norm{\delta_\alpha O_i}^2 
     \leq 
     \sum_{j=1}^N \sum_{j'=1}^N \abs{\omega_{ij'}-n_i\alpha_{j'}}
    \abs{\omega_{ij}-n_i\alpha_j}\theta_{jj'}^2/2
    =n_i^2 Q_{i, n_i\alpha}$, 
as desired.}
\end{proof}}

\del{Using Lemma~\ref{old:lem:deltaOi}, we can prove the following preliminary error-bound for proving the theorem.}

\del{\begin{proposition}
\del{Let $i \in \{1, \dotsc, N\}$. {Assume that \eqref{eq:OIneq0} holds.} Define}
\begin{equation*}
    \del{w_i = \eta A_i+\xi_i u_i.}
\end{equation*}
\del{If $\langle V \rangle_\alpha \neq 0$, then the norm of the vector}
\begin{equation*}
\del{\zeta_i
= 
\phi(D_i) 
- 
\left(
\langle V\rangle_\alpha
+ 
\frac{w_i}{n_i}
-
\frac{w_i^\top \langle V\rangle_\alpha}{n_i\norm{\langle V\rangle_\alpha}^2}\langle V\rangle_\alpha
\right)}
\end{equation*}
\del{satisfies }
\begin{equation*}
\begin{multlined}[.85\linewidth]
\del{\norm{\zeta_i}
\leq 2\tilde \pi_\alpha \rho_i + C_\phi \rho_i^2 + \tilde \pi_i(1+4\rho_i) + (1 + 2\rho_i)\sqrt{Q_{i, n_i \alpha}}.}
\end{multlined}
\end{equation*}
\end{proposition}}

\del{\begin{proof}
\del{Because $D_i = O_i + w_i$, we can decompose the vector~$\zeta_i$ as $\zeta_i = \zeta_{i, 1} + \zeta_{i, 2} + \zeta_{i, 3} + \zeta_{i, 4} + \zeta_{i, 5}+ \zeta_{i, 6}$, where }
\begin{equation*}
    \begin{aligned}
    \del{\zeta_{i, 1} }&\del{= \delta\phi(O_i, w_i), }
    \\
    \del{\zeta_{i, 2}} &\del{= \phi(O_i) - \langle V \rangle_\alpha, }
    \\
    \del{\zeta_{i, 3} }&\del{= \bigl({\norm{O_i}^{-1}-n_i^{-1}}\bigr) w_i, }
    \\
    \del{\zeta_{i, 4} }&\del{= - \bigl(\norm{O_i}^{-3}-n_i^{-3})(w_i^\top O_i\bigr) O_i, }
    \\
    \del{\zeta_{i, 5}} &\del{= \frac{w_i^\top \langle V\rangle_\alpha}{n_i\norm{\langle V\rangle_\alpha}^2}\langle V\rangle_\alpha
    - 
    \frac{w_i^\top \langle V\rangle_\alpha}{n_i}\langle V\rangle_\alpha, }
    \\
    \del{\zeta_{i, 6}} &\del{= \frac{w_i^\top \langle V\rangle_\alpha}{n_i}\langle V\rangle_\alpha
    - 
    \frac{w_i^\top O_i}{n_i^3}O_i.}
    \end{aligned}
\end{equation*}
\del{Let us evaluate the norm of each of these vectors. Lemma~\ref{old:lem:2} shows}
\begin{equation*}
\begin{aligned}
    \del{\norm{\zeta_{i, 1}} }
    &\del{\leq C_\phi \norm{O_i}^{-2} \norm{w_i}^2
    =C_\phi \rho_i^2.}
    \end{aligned}
\end{equation*}
\del{Because {$\zeta_{i, 2} = (\norm{O_i}^{-1}-n_i^{-1})O_i - n_i^{-1}\delta_\alpha O_i$}, inequality \mbox{\eqref{old:eq:deltaOiineq}} shows}
\begin{equation*}
    \begin{aligned}
    \del{\norm{\zeta_{i, 2}} }
    &\del{\leq 
    (\norm{O_i}^{-1}-n_i^{-1})\norm{O_i} + \sqrt{Q_{i, n_i \alpha}}}
    \\
    &\del{=
    \tilde \pi_i + \sqrt{Q_{i, n_i \alpha}}.}
    \end{aligned}
\end{equation*}
\del{Similarly, we can show }
\begin{equation*}
\begin{aligned}
    \del{\norm{\zeta_{i, 3}} }
    &\del{\leq ({\norm{O_i}^{-1}-n_i^{-1}})\norm{w_i}
    = \tilde \pi_i \rho_i, }
    \end{aligned}
\end{equation*}
\del{and}
\begin{equation*}
\begin{aligned}
    \del{\norm{\zeta_{i, 4}} }
    &\del{\leq
    (\norm{O_i}^{-3}-n_i^{-3})\norm{w_i} \norm{O_i}^2}
    \\
    &\del{=
    \rho_i (1-\pi_i^3)}
    \\
    &\del{\leq 
    3 \tilde \pi_i \rho_i.}
\end{aligned}
\end{equation*}
\del{Furthermore, it follows that  }
\begin{equation*}
\begin{aligned}
    \del{\norm{\zeta_{i, 5}}}
    &\del{\leq
    (1-\norm{\langle V\rangle_\alpha}^2)\norm{w_i}n_i^{-1}}
    \\
    &\del{\leq 2\tilde \pi_\alpha \rho_i.}
\end{aligned}
\end{equation*}
\del{Also, by using the triangle inequality and inequality~\eqref{old:eq:deltaOiineq}, we can show }
\begin{equation*}
\begin{aligned}
    \del{\norm{\zeta_{i, 6}} }
    &\del{\leq  2\norm{w_i}n_i^{-1} \sqrt{Q_{i, n_i \alpha}}}
    \\
    &\del{\leq 
    2\rho_i \sqrt{Q_{i, n_i \alpha}}.}
\end{aligned}
\end{equation*}
\del{Finally, inequalities
\mbox{\cref{old:eq:zeta1,old:eq:zeta2,old:eq:zeta3,old:eq:zeta4,old:eq:zeta5,old:eq:zeta6}}
complete the proof.}
\end{proof}}

\del{We are now ready to prove the main result of this section.}
\del{\begin{proof}
\del{Let us write $w = A+u$. From the definition \eqref{old:eq:tildeD} of the vector~$w_i$, we have $\sum_{i=1}^N ({\alpha_i}/{n_i}) w_i = w$. Therefore, taking the weighted summation with respect to $i$ in equation~\eqref{old:eq:zetai} shows} 
\begin{equation*}
\begin{aligned} 
   \del{\langle \zeta\rangle_\alpha}
    &\del{=}
    \begin{multlined}[t]
        \del{\langle \phi(D)\rangle_\alpha 
    -
    \left(
    \langle V\rangle_\alpha + w - \frac{w^\top\langle V\rangle_\alpha}{\norm{\langle V\rangle_\alpha}^2}\langle V\rangle_\alpha
    \right)}
    \end{multlined}
    \\
    &\del{=}
    \begin{multlined}[t]
    \del{\langle \phi(D)\rangle_\alpha -
    \norm{\langle V\rangle_\alpha}\,\phi\bigl(\langle V\rangle_\alpha + w\bigr) + \norm{\langle V\rangle_\alpha}\, \delta\phi\bigl(\langle V\rangle_\alpha, w\bigr), }
    \end{multlined}
\end{aligned}
\end{equation*}
\del{where in the last equation we used \eqref{old:eq:deltaphi}. Therefore,}
\begin{equation*}
    \del{\epsilon_\alpha = \langle \zeta\rangle_\alpha-
    \norm{\langle V\rangle_\alpha}\,\delta\phi\bigl(\langle V\rangle_\alpha, w\bigr).}
\end{equation*}
\del{Hence, the triangle inequality as well as 
Proposition~\ref{old:prop:zeta} and Lemma~\ref{old:lem:2} complete the proof of inequality~\eqref{eq:error:alpha}.}
\end{proof}}

\del{
\setcounter{subsection}{2}
\subsection{\del{Prediction Model}}
\setcounter{subsection}{1}
\addtocounter{subsubsection}{1}
}

\subsubsection{\add{Prediction Model}}\label{subsec:staticPrediction}

 
{Our discussion at the beginning of this section and} Theorem~\ref{thm:static} \del{suggests}{suggest} the approximation 
\begin{equation}\label{eq:suggestedApproximation}
     \langle V(k+1) \rangle_{\del{\alpha}{1/N}} \approx \norm{\langle V(k)\rangle_{\del{\alpha}{1/N}}}\,\phi\bigl(\langle V(k)\rangle_{\del{\alpha}{1/N}} + A(k) + u(k)\bigr). 
\end{equation}
This approximation motivates us to construct the following lower-order prediction model, which we denote by $\Sigma_{\stc}$. In this model, at an observation time~$k = \ell T$, the estimated center of mass $\hat c$ is initialized as 
\begin{equation}\label{eq:hatxinitialize}
    \hat c(\ell T) = \frac{1}{N}\sum_{i=1}^N x_i(\ell T), 
\end{equation}
and \add{it} is dynamically updated by the difference equation 
\begin{equation}\label{eq:hatxupdate}
    \hat c(k+1) = \hat c(k) + \tau v \hat V(k)\add{,}
\end{equation}
where $\hat V$ represents the estimated direction of the center of mass\del{,} and is dynamically updated as in \eqref{eq:suggestedApproximation} by\del{ the equation}
\begin{equation}\label{eq:isreferred_hatvupdate}
\begin{multlined}[.85\linewidth]
    \hat V (k+1) =
    \norm{\hat V(k)}\,\phi\bigg( \hat V(k) +  \sum_{i=1}^N \frac{\del{N^{-1}}{1/N}}{n_i(\ell T)} \eta A_i(\ell T) +  \sum_{i=1}^N \frac{\del{N^{-1}}{1/N} }{n_i(\ell T)}\xi_i u_i(k) \bigg)
    \end{multlined}
\end{equation}
for $k \geq \ell T$ under the initialization 
\begin{equation}\label{eq:initialCond}
    \hat V(\ell  T) = \frac{1}{N}\sum_{i=1}^N V_i(\ell  T)
\end{equation}
at time $k=\ell T$. 
{Therefore, the optimization problem to be solved during the continuous-time interval $[\ell T \tau, (\ell+1) T \tau)$ \del{for determining}\add{to determine} the stimulus to be used in the next interval $[(\ell+1) T \tau, (\ell+2) T \tau)$ is}
\begin{equation}\label{eq:staticWeightMPCOptimization}{
    \begin{aligned}
        \minimize_{\{u(k)\}_{k=(\ell +1)T}^{(\ell +1)T+T_h-1}}& 
    \sum_{k=(\ell+1)T}^{(\ell+1)T+T_h}\dist(\hat c(k), \mathcal R)
\\
\subjectto & \mbox{\cref{eq:hatxinitialize,eq:hatxupdate,eq:isreferred_hatvupdate,eq:initialCond}.}
    \end{aligned}}
\end{equation}
\add{A schematic of the control algorithm is illustrated in Fig.~\ref{fig:MPCstc}.} \label{whatquantities}\add{The quantities {required} to solve this optimization problem are the observations~$\{x_i(\ell T), V_i(\ell T)\}_{i=1}^N$ and the parameters~$\psi$, $\eta$, $\xi_i$, $N$, $\tau$, $v$, $r_r$, $r_o$, and~$r_a$. If the exact values of the parameters are not available, they can be replaced with estimates.}

\begin{figure}[tb]
    \centering
    \add{\includegraphics[width=.75\linewidth]{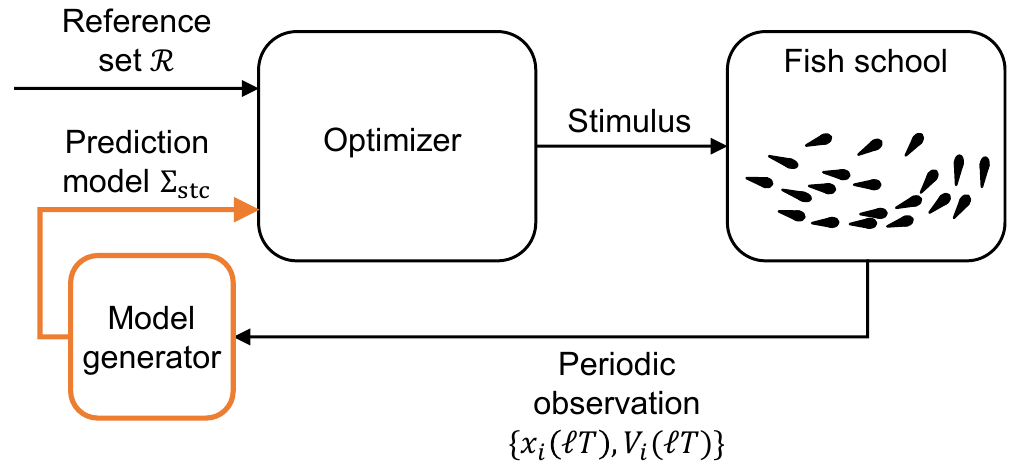}}
    \caption{\add{A schematic of MPC with the reduced-order model~$\Sigma_{\stc}$. The positions and directions of each individual are measured periodically. Using the measurement, we construct reduced-order model~$\Sigma_{\stc}$ for prediction. Using the model in place of~$\Sigma_{\orig}$, the  optimizer periodically updates the stimulus applied to individuals.%
    }}
    \label{fig:MPCstc}
\end{figure}


\del{
\section{\del{Dynamic weight}} 
\addtocounter{section}{-1}
\setcounter{subsection}{1}
}
\subsection{\add{Dynamic Weight}}
\label{sec:reduction:dynamic}

\del{In the last section, we have}\add{We} closely analyzed the error from a dimension reduction of the controlled fish school model\del{,} and \del{then }introduced a reduced-order prediction model for the schooling. In the reduction, the weights for aggregating the \del{fish's }directions \add{of the fish} in the school were assumed to be constant. Then, \del{a}\add{the} natural question is to ask whether allowing the weights to \emph{\del{dynamically }change} \add{dynamically} over time dependent on the state of the school would lead to a more precise reduction of the model.

To provide an affirmative answer to the aforementioned question, in this section, we present a dynamic weight-reduction of the controlled fish schooling model and, then, present yet another model for predicting the mass center of the fish school. We \del{specifically }show that the reduction weighted by the \emph{normalized eigenvector centrality} of the interaction network induced by forces of orientation \del{is potentially more accurate than the reduction with static weights}{allows us to improve the error bound~\eqref{eq:error:alpha} for \del{the case of }static weights}.

\del{To further proceed, let}\add{Let} us recall the definition of \add{the} normalized eigenvector centrality of directed networks. Let $\mathcal G = (V, E)$ be a directed and unweighted graph having nodes $ V = \{1, \dotsc, N\}$. We denote \del{by $(i, j)$ }an directed edge from \del{the }node $i$ to \del{the }node $j$ \add{by $(i, j)$}. For each $i \in V$, let $d_{\out}(i)$ denote the outdegree of \del{the }node $i$. The normalized adjacency matrix of~$\mathcal G$ is the $N\times N$ matrix defined by 
\begin{equation*}
    W_{ij} = 
    \begin{cases}
    {1}/{d_{\out}(i)},&\mbox{if $(i, j) \in E$,}
    \\
    0,&\mbox{otherwise.}
    \end{cases}
\end{equation*}
If $\mathcal G$ is strongly connected, \del{then }the normalized eigenvector centrality of~$\mathcal G$ is defined as the positive vector~$\beta = [\beta_1\ \cdots\ \beta_N]^\top$ satisfying $W^\top \beta = \beta$ and~$\beta_1 + \cdots + \beta_N = 1$. 

\del{
\setcounter{section}{5}
\setcounter{subsection}{0}
\subsection{\del{Reduction Error Analysis}}
\addtocounter{subsection}{-1}
\setcounter{section}{4}
\setcounter{subsection}{2}
}
\subsubsection{\add{Reduction Error Analysis}}

{\del{Before presenting our main result of this section, we give}\add{We first provide} a brief discussion \add{to demonstrate the need for using} \del{for motivating the use of }the normalized eigenvector centrality of the interaction network induced by forces of orientation. \del{For this purpose, we consider the situation}\add{We consider the scenario} where the orientation force $O_i(k)$ is a dominant term in the desired direction~$D_i(k)$. In this \del{situation}\add{scenario}, we can invoke the rough approximation $\norm{D_i(k)} \approx n_i(k)$, which leads to}
\begin{equation}\label{eq:Diapprox}{
    D_i(k) \approx \frac{1}{n_i}(O_i + \eta A_i + \xi_i w_i).}
\end{equation}

\begin{figure}[tb]
    \centering
    \add{\includegraphics[width=.4\linewidth]{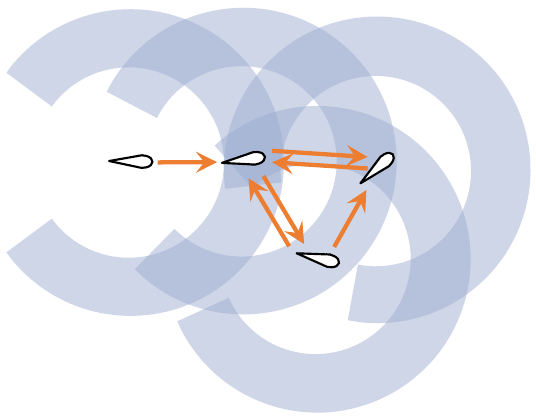}}
    \caption{\add{Interaction network $\mathcal G_o$ of a fish school with $N=4$ individuals. Blue areas represent the regions of orientation of each individuals.}}
    \label{fig:intnet}
\end{figure}

Now, \del{let us}\add{we} define the interaction network $\mathcal G_o(k)$ induced by forces of orientation as the directed and unweighted network  in which nodes are fish in the school and edge $(i, j)$ exists if and only if $j \in \mathcal N_{o, i}(k)$ \add{(Fig.~\ref{fig:intnet})}. Let $W_o(k)$ denote the normalized adjacency matrix of the network $\mathcal G_o(k)$. Then, \del{we have}\add{a straightforward calculation shows}
\begin{equation}\label{eq:intuition}
    \begin{aligned}
    \sum_{i=1}^N \alpha_i \frac{1}{n_i(k)} O_i(k)
    &=
    \sum_{i=1}^N \alpha_i \frac{1}{n_i(k)}\sum_{j=1}^N \omega_{ij}(k) V_j(k)
    \\
    &=
    \langle V(k)\rangle_{W_o(k)^\top\alpha}. 
    \end{aligned}
\end{equation}
Therefore, from \eqref{eq:Diapprox}\add{,} we obtain 
\begin{equation*}
    \langle D(k)\rangle_\alpha \approx \langle V(k)\rangle_{W_o(k)^\top \alpha} 
    + A(k) + u(k), 
\end{equation*}
where $A(k)$ and~$u(k)$ are defined in \eqref{eq:def}. This approximate relationship contains a potential inconsistency in the weights on the desired direction~$D$ and the direction $V$. The inconsistency can be resolved if \del{the }weight~$\alpha$ satisfies $\alpha = W_o(k)^\top \alpha$, i.e., if $\alpha$ equals the normalized eigenvector centrality of \del{the }network~$\mathcal G_o(k)$. This discussion suggests that the aggregation of \del{the }direction~$V$ by the normalized eigenvector centrality \del{could}\add{can} allow us to construct a reduced-order model with a higher accuracy.

The following theorem provides a theoretical confirmation on the potential superiority of \del{the }model reduction with the normalized eigenvector centrality over the one with a static weight. The dependence of relevant variables on the time index~$k$ is omitted for brevity\add{,} as in Theorem~\ref{thm:static}.  

\begin{theorem}\label{thm:nec}
Assume that \eqref{eq:empty} and \eqref{eq:OIneq0} hold true for all $i=1, \dotsc, N$. \del{If}{Suppose that} $\alpha$ equals the normalized eigenvector centrality of~$\mathcal G_o${. If $\langle V \rangle_\alpha \neq 0$}, then the vector~$\epsilon_\alpha$ defined in~\eqref{eq:eps} satisfies 
\begin{equation}\label{eq:alphastarineq}
\begin{multlined}[.85\linewidth]
    \norm{\epsilon_\alpha} 
    \leq 
    \del{C_\phi}\add{2}
    \frac{\norm{A+u}^2}{\norm{\langle V\rangle_\alpha}}+
    \sum_{i=1}^N
    \alpha_i\tilde \pi_i 
    +
    \add{2}\sum_{i=1}^N
    \alpha_i\rho_i\left(\del{2}\tilde \pi_\alpha + \del{C_\phi} \rho_i +\del{4}\add{2} \tilde \pi_i + \del{2} \sqrt{Q_{i, n_i \alpha}}\right).
    \end{multlined}
\end{equation}
\end{theorem}

\begin{proof}
\add{See \ref{sec:pfthmnec2}.}
\end{proof}

\del{Notice that the}\add{The} bound of the reduction error in \eqref{eq:alphastarineq} does not have the term $\sum_{i=1}^N \alpha_i \sqrt{Q_{i, n_i\alpha}}$, which is \del{one of the}\add{a} major \del{terms}\add{term} in \del{the }\del{estimate~\eqref{eq:error:alpha}}{estimates~\eqref{eq:error:alpha} and~\eqref{eq:stat:au0}} from using static weights. This fact suggests that using the normalized eigenvector centrality as the weight for reduction \del{could}\add{can} lead to a more precise prediction of \del{the }fish schooling. 

{The following corollary is the dynamic-weight counterpart of Corollary~\ref{cor:static} and explicitly shows the improvement of \del{the }estimate~\eqref{eq:stat:au0}. 
\begin{corollary}
Assume that \eqref{eq:empty} and \eqref{eq:OIneq0} hold true for all $i=1, \dotsc, N$. Suppose that $\alpha$ equals the normalized eigenvector centrality of~$\mathcal G_o$. If $A_i=u_i=0$ for all $i\in \{1, \dotsc, N\}$, then  
\begin{equation*}
    \norm{\epsilon_\alpha} 
    \leq 
    \sum_{i=1}^N
    \alpha_i\tilde \pi_i . 
\end{equation*}
On the other hand, if $\theta_{ij} = 0$ for all $i, j\in\{1, \dotsc, N\}$, then inequality~\eqref{eq:stat:theta0ineq} holds. 
\end{corollary}}

\del{Let us give the proof of Theorem~\ref{thm:nec}.}

\del{\begin{proof}[Proof of Theorem~\ref{thm:nec}] 
\del{Define the vector~$\zeta_i$ as in \eqref{eq:zetai}. Assume that $\alpha$ equals the normalized eigenvector centrality of~$\mathcal G_o(k)$. Let us decompose the vector~$\langle\zeta \rangle_\alpha = \sum_{i=1}^N \alpha_i \zeta_i$ as  $\langle\zeta \rangle_\alpha = \omega_{1} + \omega_{2} + \omega_{3} + \omega_{4} + \omega_{5} + \omega_6$, where $\omega_\ell = \sum_{i=1}^N \alpha_i \zeta_{i, \ell}$ for all $\ell=1,\dotsc, 6$. We bound the norm of the vectors $\omega_1$, $\omega_3$, $\omega_4$, $\omega_5$, and~$\omega_6$ by using inequalities \eqref{old:eq:zeta1} and \eqref{old:eq:zeta3}--\eqref{old:eq:zeta6}. Let us evaluate $\norm{\omega_2}$ in a different manner. Because $\alpha$ is assumed to be equal to the normalized eigenvector centrality of~$\mathcal G_o(k)$, from equation \eqref{eq:intuition} we can show 
that}
\begin{equation*}
\del{\omega_2 
= 
\sum_{i=1}^N \alpha_i(\phi(O_i) - \langle V \rangle_\alpha) 
=  \sum_{i=1}^N \alpha_i (\phi(O_i) - n_i^{-1} O_i).}
\end{equation*}
\del{Therefore, we obtain $\norm{\omega_2} \leq \sum_{i=1}^N \alpha_i \tilde \pi_i$. This inequality as well as inequalities \eqref{eq:old:zeta1} and  \eqref{eq:old:zeta3}--\eqref{eq:old:zeta6} show}
\begin{equation*}
\begin{multlined}[.85\linewidth]
\del{\norm{\langle \zeta\rangle_\alpha}
 \leq \sum_{i=1}^N \alpha_i\bigg(2\tilde \pi_\alpha \rho_i + C_\phi \rho_i^2 
 + \tilde \pi_i(1+4\rho_i) + 2\rho_i\sqrt{Q_{i, n_i \alpha}}\bigg).} 
 \end{multlined}
\end{equation*}
\del{Therefore, equality~\eqref{eq:old:finalepailonalpha} as well as 
Proposition~\ref{prop:zeta} and Lemma~\ref{lem:2} complete the proof of the theorem, as desired.}
\end{proof}}

\del{
\setcounter{section}{5}
\setcounter{subsection}{1}
\subsection{\del{Prediction Model}}
\setcounter{section}{4}
\setcounter{subsection}{2}
\setcounter{subsubsection}{1}
}

\subsubsection{\add{Prediction Model}}\label{subsec:dynpremod}

We present our second reduced-order model~\del{$\hat \Sigma_{\dyn}$}\add{$\Sigma_{\dyn}$} based on Theorem~\ref{thm:nec}. In this model, although the estimated center of mass~$\hat c$ is initialized and updated by equations~\eqref{eq:hatxinitialize} and~\eqref{eq:hatxupdate} as in the first model~{$\hat \Sigma_{\stc}$}, the update law for the estimated direction $\hat V$ is distinct and \del{is }given by 
\begin{equation}\label{eq:isreferred_hatvupdateddyn}
\begin{multlined}
    \hat V (k+1) = \norm{\hat V (k) }\, \phi\bigg( \hat V(k) + \sum_{i=1}^N \frac{\beta_i(\ell  T)}{n_i(\ell  T)} \eta A_i(\ell  T) + \sum_{i=1}^N \frac{\beta_i(\ell  T)}{n_i(\ell  T)}\xi_i u_i(k) \bigg) 
    \end{multlined}
\end{equation}
for $k = \ell T, \ell T+1, \dotsc$, where \del{$\beta_i(\ell T)$}\add{$\beta_1(\ell T)$, \dots, $\beta_N(\ell T)$} \del{denotes}\add{denote} the normalized eigenvector centrality of the network $\mathcal G_o(\ell T)$. \del{We remark that, in}\add{In} this construction, we implicitly utilize the approximation $\langle V \rangle_\alpha \approx \langle V \rangle_{1/N}$. Therefore, we use the same initialization~\eqref{eq:initialCond} {of~$\hat V$} as in~$\Sigma_{\stc}$. Hence, the optimization problem to be solved during the continuous-time interval $[\ell T \tau, (\ell+1) T \tau)$ for determining the stimulus to be used in the next interval~$[(\ell+1) T \tau, (\ell+2) T \tau)$ is \begin{equation}\label{eq:dynamicWeightMPCOptimization}
    \begin{aligned}
        \minimize_{\{u(k)\}_{k=(\ell +1)T}^{(\ell +1)T+T_h-1}}& 
    \sum_{k=(\ell+1)T}^{(\ell+1)T+T_h}\dist(\hat c(k), \mathcal R)
\\
\subjectto & \cref{eq:hatxinitialize,eq:hatxupdate,eq:isreferred_hatvupdateddyn,eq:initialCond}.
    \end{aligned}
\end{equation}
\label{whatquantitiesdyn}\add{Although the prediction model~$\Sigma_{\dyn}$ is more elaborated than $\Sigma_{\stc}$, the set of  quantities we require to solve this optimization problem is the same as the one for the optimization problem~\eqref{eq:staticWeightMPCOptimization}.}

\del{Because }Theorem~\ref{thm:nec} suggests {the possibility} that \del{the }model~$\Sigma_{\dyn}$ {achieves} smaller prediction errors compared with the static weight-based model~$\Sigma_{\stc}$\del{,} \add{because} we can expect that the MPC with the model~$\Sigma_{\dyn}$ can be better than the one with~$\Sigma_{\stc}$. We \del{affirmatively }confirm the superiority of the prediction model~$\Sigma_{\dyn}$ in the next section. 


\begin{figure*}[bt]
\centering
\includegraphics{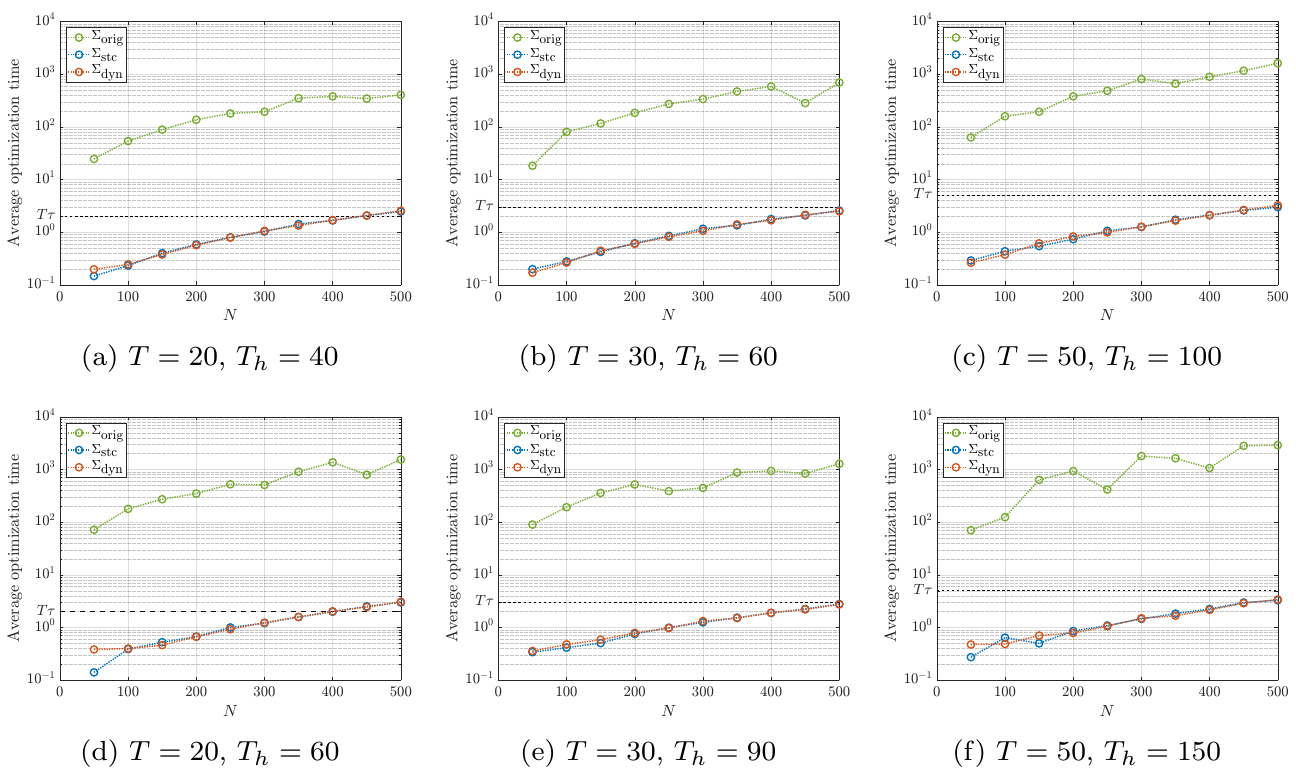}
\caption{Average of optimization times in one instance of MPC with prediction models $\Sigma_{\orig}$, $\Sigma_{\stc}$, and~$\Sigma_{\dyn}$. The reference set has radius $r_{{\reference}} = 2000$. \add{The optimization with the raw prediction model~$\Sigma_{\orig}$ does not complete within the required $T\tau$ seconds in any of the cases, whereas the MPC with the reduced-order prediction models is both successful for fish schools with up to $N=350$ individuals for any pair~$(T, T_h)$.}}
\label{fig:computationTime}
\end{figure*}

\section{Numerical Simulation}\label{sec:num}

In this section, we present numerical simulations and illustrate the effectiveness of the MPC with the reduced-order prediction models. \del{We also}\add{Further, we} numerically show that the dynamic-weight prediction model allows us to perform better tracking than \add{that with} the static-weight prediction model.  \del{The simulations}\add{Simulations} have been performed with MATLAB 2021a on a desktop computer with Xeon Gold 5218R processors and 96~GB DRAMs. \del{For solving the optimization problems \eqref{eq:clarifiedoptimizationproblem}, \eqref{eq:staticWeightMPCOptimization}, and \eqref{eq:dynamicWeightMPCOptimization}, we have}\add{We} used the nonlinear programming solver of \texttt{fmincon} \add{to solve the optimization problems~\eqref{eq:clarifiedoptimizationproblem}, \eqref{eq:staticWeightMPCOptimization}, and \eqref{eq:dynamicWeightMPCOptimization}}.

{In this numerical simulation, we assume that the fish school consists of homogeneous individuals with the same sensitivity; i.e., we suppose the existence of a positive constant $\xi$ such that $\xi_1 = \cdots = \xi_N =\xi$.} We \del{also}\add{further} assume, for the purpose of simplicity and illustration, that the same control input is uniformly applied to \del{the }all individuals in the school. Therefore, we assume $u_1(k) = \cdots = u_N(k)$ for all $k$. We follow the reference~\cite{Gautrais2008} and set the relevant parameters of the fish school as $v = 50$, $\psi = \del{3\pi/2}{3\pi/4}$, $\theta = \del{40(\pi/180)}\add{0.69}$ [rad], $\eta=1$, $r_r = 50$, $r_o = 750$, and~$r_a = \num[group-separator={,}]{1000}$. We set the noise level to be 10 times of the one suggested by the authors in~\cite{Gautrais2008} to make the tracking harder. We set the sampling interval for discretization to be $\tau = 0.1$\del{. The}\add{; the} common sensitivity parameter is set as $\xi = 10$. 

We suppose that \del{the }reference set~$\mathcal R$ is the surface of a sphere centered at the origin \del{and having}\add{with} radius $r_{{\reference}}$. To determine the initial positions of fish in the school, we first place each fish by drawing from the uniform distribution in a sphere with radius~$r_a/2$. Then, the entire set of fish is uniformly shifted so that the mass center $c(0)$ of the school is \del{exactly }located on a point on the reference set~$\mathcal R$. Our choice of the radius~$r_a/2$ guarantees that the initial fish school is \add{sufficiently} large \del{enough }so that both \del{the }orientation~\eqref{eq:def:O_i} and the attraction~\eqref{eq:def:A_i} forces are in effect, while ensuring that each pair of fish is within the interaction distance to avoid the fragmentation of the school.

\begin{figure*}[b]
\centering
\vspace{-6mm}
\includegraphics[width=.82\linewidth]{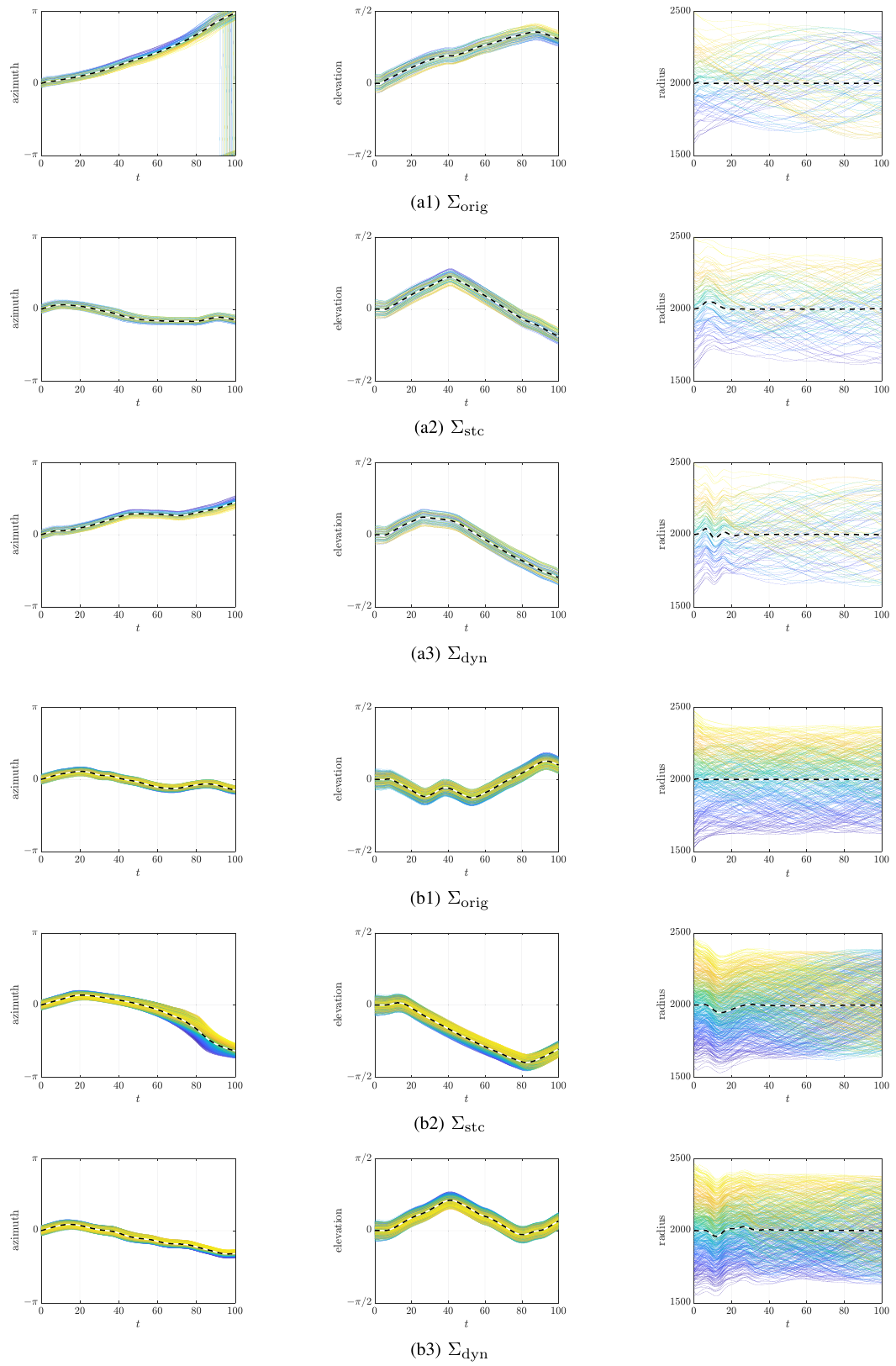}
\caption{Azimuth (left), elevation (middle), and radius (right) of positions $x_i(k)$ of individual fish (solid lines) {controlled with prediction models $\Sigma_{\orig}$, $\Sigma_{\stc}$, and~$\Sigma_{\dyn}$.} \del{and }{Dashed line represent} \add{the} mass center~$c(k)$ \del{(dashed line)}{of fish school}. \del{Reference}\add{The reference} set has radius $r_{{\reference}} = 2000$. {a1-a3: $N=100$. b1-b3: $N=300$.}}
\label{samplePaths300}
\end{figure*}

\begin{figure*}[tb]
\centering
\includegraphics{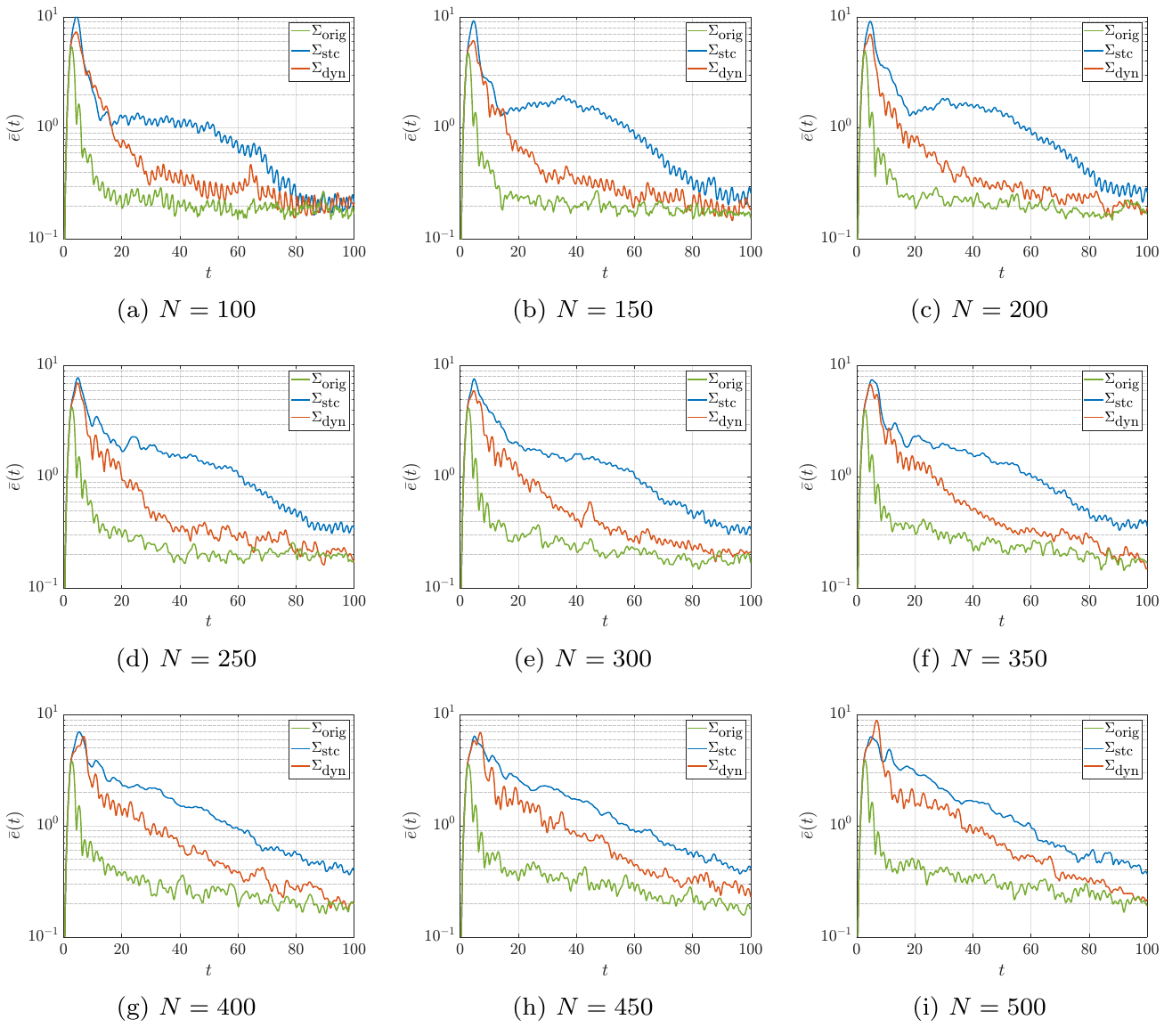}
\caption{Average tracking errors in the MPC with {prediction models $\Sigma_{\orig}$ (green),} \del{static-weight model }$\Sigma_{\stc}$ (blue){,} and \del{dynamic-weight model }$\Sigma_{\dyn}$ (red). \del{Reference}\add{The reference} set has radius $r_{{\reference}} = 2000$. \add{Overall, the MPC with the prediction model~$\Sigma_{\dyn}$ performs better than the one with \del{the }model~$\Sigma_{\stc}$ using static weights. Furthermore, the MPC with the prediction model~$\Sigma_{\dyn}$ performs asymptotically as good as the one with the original prediction model~$\Sigma_{\orig}$.}}
\label{averageerror}
\end{figure*}

\subsection{Computation Time}

We \del{first }investigate the amount of \del{the }reduction in computation time \del{for performing}\add{required to perform} the MPC. Let \del{us set }the radius of the reference sphere \del{to }be $r_{{\reference}}=2r_a= \num[group-separator={,}]{2000}$. We examine the following three different scenarios on \del{the }period~$T$ of observation and actuation: $T=20$, $30$, and~$50$. For each $T$, we \del{further }consider two different \del{length}\add{lengths} of the optimization horizon $T_h$: $T_h = 2T$ and~$T_h= 3T$. For \del{each of the total 6}\add{all six} cases, we perform the following simulation. We randomly set the initial positions and directions of schools of fish with~$N=50$, $100$, $150$, \dots, $500$ individuals. We then perform one instance of MPC \del{by }using the three prediction models $\Sigma_{\orig}$, $\Sigma_{\stc}$\add{,} and~$\Sigma_{\dyn}$ on the discrete-time interval $\{0, 1, \dotsc, 1000\}$ (i.e., the continuous-time interval~$[0, 100]$). \del{We then}\add{Then, we} measure the average of the computation time of finite-horizon optimizations.  

In Fig.~\ref{fig:computationTime}, we \del{show}\add{present} the average of the computation time of finite-horizon optimizations with each of the prediction models {(i.e., the optimization problems in~\eqref{eq:clarifiedoptimizationproblem}, \eqref{eq:staticWeightMPCOptimization}, and~\eqref{eq:dynamicWeightMPCOptimization})} and for the six pairs~$(T, T_h)$. \del{The optimization}\add{Optimization} with the raw prediction model~$\Sigma_{\orig}$ does not complete within the required $T\tau$ seconds in any \del{of the cases}\add{case}. \del{On the other hand, the}\add{The} MPC with the reduced-order prediction models is both successful for fish schools with up to $N=350$ individuals for any \del{of the cases}\add{case}, which confirms the computational effectiveness of \del{the }model reduction. 

\begin{figure*}[bt]
\centering
\includegraphics{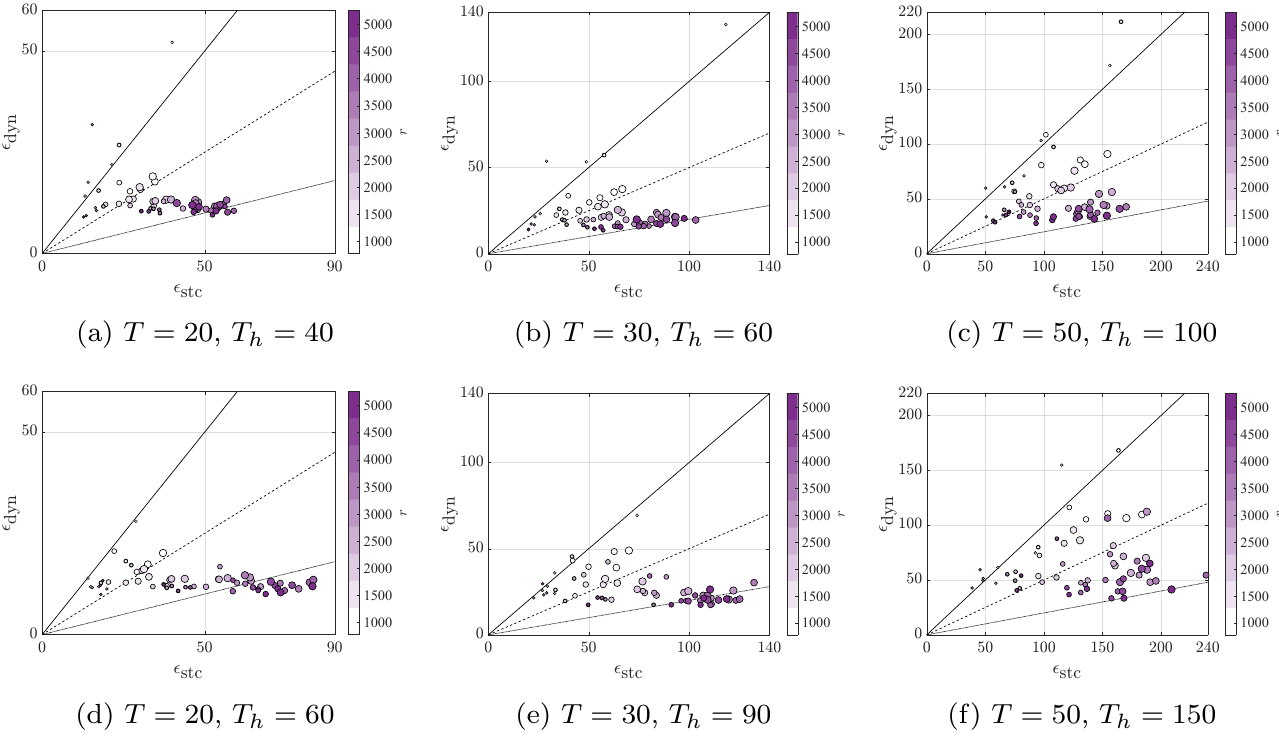}
\caption{Asymptotic cumulative tracking errors $\epsilon_{\stc}$ and~$\epsilon_{\dyn}$ from MPC with static and dynamic weight-prediction models, respectively. Each circle represents data from performing 100 control trials for a fixed pair $(N, r_{{\reference}})$. Area of each circle is proportional to $N$. Colors of markers represent $r_{{\reference}}$. Solid, dashed, and dotted lines represent the straight lines $\epsilon_{\dyn} = \epsilon_{\stc}$, $\epsilon_{\dyn} = \epsilon_{\stc}/2$, and~$\epsilon_{\dyn} = \epsilon_{\stc}/5$, respectively. \add{Overall, the errors from the dynamic weight-model is smaller than the ones from the static weight-model for most of the cases.}}
\label{animals}
\end{figure*}

\subsection{Tracking Performance}\label{subsec:tp}

\del{We then compare the tracking performances. In the remaining of the paper, we focus on the model predictive control with the reduced-order prediction models~$\Sigma_{\stc}$ and~$\Sigma_{\dyn}$ because the one with the raw prediction model~$\Sigma_{\orig}$ does not complete within the required time as we have observed in the previous subsection.} 

{We \del{then }compare the tracking performances.} \del{We first}\add{To this end, we} fix $r_{{\reference}}=2000$ and perform \add{the} MPC of the fish school consisting of~$N=100$ and~$300$ individuals. The \del{results of the simulations}\add{simulation results} are illustrated in Fig.~\ref{samplePaths300}. \del{In the figures, we}\add{We} show the azimuth, elevation, and radius of the position of all individuals and the mass center. For both \del{of the }school sizes, we \del{can }confirm that the MPC with {all} \del{both the }prediction models\del{$\Sigma_{\stc}$ and~$\Sigma_{\dyn}$} realize stable tracking to the sphere. 

\del{In order to further compare the tracking performance of the MPC with the {three} prediction models, we}\add{We} vary the school sizes and observe the average tracking errors \add{to further compare the tracking performance of the MPC with the {three} prediction models}. \del{Specifically, for each school size~$N$, we}\add{We} perform 100 instances of MPC over the same time interval with different initial positions and directions \add{for each school size~$N$}. \del{We then calculate}\add{Then, we calculate} the empirical average~$\bar e$ of the tracking error~$e$ defined in \eqref{eq:defe}. The obtained empirical averages of the tracking errors are presented in Fig.~\ref{averageerror}. We observe an overall trend in which the MPC with \del{the prediction model }$\Sigma_{\dyn}$ performs better than the one with \del{the model }$\Sigma_{\stc}$ using static weights. {We also find that the MPC with \del{the prediction model }$\Sigma_{\dyn}$ performs asymptotically as good as the one with \del{the original prediction model }$\Sigma_{\orig}$.} {These trends suggest} the effectiveness of dynamically changing the weights \del{in}\add{on} model reduction, which \del{thereby }confirms our theoretical findings on the reduction errors in Theorems~\ref{thm:static} and~\ref{thm:nec}.  

{Let us perform}\del{For} a more thorough comparison of the MPC with \del{the dynamic weight-prediction model }$\Sigma_{\dyn}$ and \del{the static counterpart }$\Sigma_{\stc}${.}\del{,} \del{The reason for focusing on}\add{We focus on} these two models \del{is that}\add{because} the MPC with \del{the raw prediction model }$\Sigma_{\orig}$ does not complete within the required time as \del{we have }observed in Fig.~\ref{fig:computationTime}. \del{we}{We specifically} compare the average of the asymptotic and cumulative tracking error
\begin{equation*}
    \epsilon = \int_{50}^{100} \bar e(t)\,dt    
\end{equation*}
for various pairs $(N, r_{{\reference}})$ of the parameters. In this simulation, 
we vary $N$ within the set $\{50, 100, \dotsc, 500\}$, while we vary the radius~$r_{{\reference}}$ of the reference sphere $\mathcal R$ within the set $\{1000, 1500, \dotsc, 5000\}$. For each pair $(N, r_{{\reference}})$, we run 100 instances of MPC on the discrete-time interval $\{0, 1, 2, \dotsc, 1000\}$ using each \del{of the }prediction model\del{s}. We plot \del{the }errors $\epsilon_{\dyn}$ and~$\epsilon_{\stc}$ from the dynamic and static prediction models, respectively, in Fig.~\ref{animals}. \del{We observe that the}\add{The} errors from the dynamic weight-model \del{is}\add{are} smaller than \del{the ones}\add{those} from the static weight-model for most \del{of the }cases, \del{confirming}\add{which confirms} the effectiveness of using the dynamic weight-model for prediction. We also remark that, while $\epsilon_{\dyn}$ is not heavily dependent on the optimization horizon~$T_h$, $\epsilon_{\stc}$ becomes larger for the longer optimization horizon~$T_h = 3T$. This phenomena \del{could}\add{can} be attributed to the relative inaccuracy of the static prediction model~$\Sigma_{\stc}$. \add{Further, we  observe the trend that the larger the radius~$r_{\reference}$ of the reference set~$\mathcal R$, the smaller is the relative accuracy~$\epsilon_{\dyn}/\epsilon_{\stc}$ of the MPC with dynamic weights. From this trend, we can infer that the MPC with the dynamic weights exhibit its superiority for tracking a reference set with a small curvature. Further,  for small schools, the MPC with dynamic weights does not necessarily outperform the one with static weights.}

\add{Therefore, we have confirmed the superiority of the prediction model~$\Sigma_{\dyn}$. The confirmation is limited to a case in which parameters of the schooling model coincide with the ones in \cite{Gautrais2008}. The parameters in \cite{Gautrais2008} have a certain degree of biologically plausibility, and therefore, the superiority confirmed so far is of practical significance. However, it is important to observe how robust the superiority is with respect to changes in the model parameters. Therefore, we perform the following simulations. Within the simulations, we fix $(T, T_h) = (30, 90)$ and repeat the same set of simulations as Fig.~\ref{animals} with different system parameters. We vary each  parameter~$\xi$, $\psi$, $\theta$, $r_o$, and~$r_a$ while others are fixed. In the simulations, we consider the following ten cases: (a) $\xi = 15$, (b) $\xi=5$, (c)  $\theta= 0.92$, (d) $\theta = 0.23$, (e) $\psi = 7\pi/8$, (f) $\psi = 5\pi/8$, (g) $r_o = 875$, (h) $r_o = 625$, (i) $r_a = 1125$, and (j) $r_a= 875$. The results of the simulations are presented in Fig.~\ref{animalsanimals}. Overall, MPC with $\Sigma_{\dyn}$ is more effective than the one with $\Sigma_{\stc}$.}

\begin{figure*}[tb]
\centering
\includegraphics{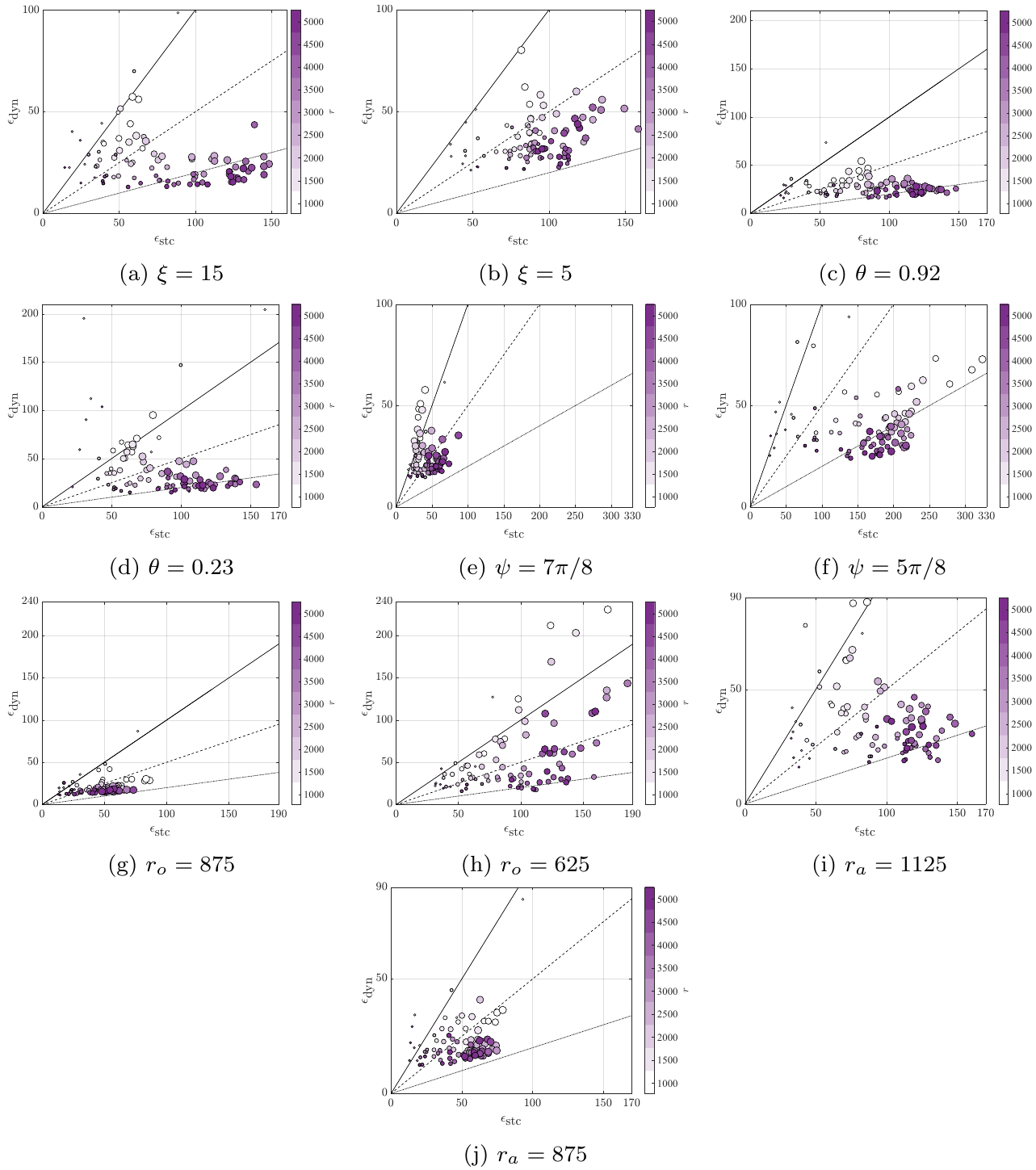}
\caption{\add{Comparison of asymptotic cumulative tracking errors $\epsilon_{\stc}$ and~$\epsilon_{\dyn}$. Each circle represents data from performing 100 control trials for a fixed pair $(N, r_{{\reference}})$. The area of each circle is proportional to $N$. The colors of markers represent $r_{{\reference}}$. Solid, dashed, and dotted lines represent the straight lines $\epsilon_{\dyn} = \epsilon_{\stc}$, $\epsilon_{\stc}/2$, and~$\epsilon_{\stc}/5$, respectively. Overall, errors from the dynamic weight-model are smaller than the ones from the static weight-model for most cases.}}
\label{animalsanimals}
\end{figure*}

\subsection{{Discussion}}\label{sec:discussion}

\add{We investigated the reduced-order MPC of the fish schooling model reported by Gautrais et al.~\cite{Gautrais2008}.  For the schooling model, we proposed reduced-order models based on the weighted average of the directions of individual fish. We then carefully analyzed the  reduction errors, which suggested that the normalized eigenvector centrality of the alignment-interaction network within the school lead to better MPC. Finally, this finding and numerical efficiency of the MPC based on reduced-order models are confirmed by extensive simulations, in which we considered various scenarios with different school sizes $N$ and the curvature of the reference set~$\mathcal R$. Therefore, this paper expands the applicability of reduced-order nonlinear MPC \cite{Wiese2015,Zhang2019c} to the fish schooling model. This paper also strengthens the existing connection \cite{Liu2012b,Fitch2016} between control and centrality measures.}

\del{We have numerically shown the superiority of the MPC with the model $\Sigma_{\dyn}$ based on the dynamic weights over the one $\Sigma_{\stc}$ based on the static weights, confirming the validity of the hypothesis drawn from our analyses in Theorems~\ref{thm:static} and~\ref{thm:nec}. This finding is supported by considering various situations with different school size $N$ and the curvature of the reference set~$\mathcal R$. Specifically, in Fig.~\ref{animals}, we can observe the trend that, the larger the radius~$r_{\reference}$ of the reference set~$\mathcal R$, the smaller the relative accuracy~$\epsilon_{\dyn}/\epsilon_{\stc}$ of the MPC with dynamic weights. From this trend, we can infer that the MPC with the dynamic weights will exhibit its superiority particularly for tracking a reference set with small curvature. We also find that, for small schools, the MPC with dynamic weights does not necessarily outperform the one with static weights.}

Although \del{the }bounds~\eqref{eq:error:alpha} and~\eqref{eq:alphastarineq} on the approximation errors suggest a robust MPC formulation~\cite{Mayne2016}, the current bounds do not directly allow us to employ such \add{a} formulation. \del{Because the error bounds \eqref{eq:error:alpha} and \eqref{eq:alphastarineq} are not closed in the variable $\langle V\rangle_\alpha$ (i.e., $\hat V$) within the prediction model, the}\add{The} bounds do not allow us to evaluate the approximation error within the whole prediction horizon if its length is longer than one \add{because error bounds \eqref{eq:error:alpha} and \eqref{eq:alphastarineq} are not closed in the variable $\langle V\rangle_\alpha$ (i.e., $\hat V$) within the prediction model}. \del{For the same reason}\add{Further}, the current error bounds do not enable us to formulate MPC with a guaranteed tracking performance. To incorporate the error bounds into the MPC, it is necessary to derive \del{an}other error bounds closed in the variable $\langle V\rangle_\alpha$.

\add{Toward future application of the proposed method for the control of actual fish \add{schools}, it is important to analyze the robustness of the method \add{against} external disturbances and model uncertainties that are common in practice. \add{For the} disturbances, random rotations incorporated into the current model can be \add{considered to represent} external disturbances \add{with} relatively short time constants. Furthermore, we confirmed the effectiveness of the proposed method for the model whose magnitude of random rotation is set as ten times large as the parameter adopted in the original model~\cite{Gautrais2008}. Therefore, the proposed method has a certain degree of robustness \add{against} external disturbances with small time constants. On the other hand, analyzing robustness of the proposed method to external disturbances with long time constants, such as unexpected steady flow in water, remains an open problem. Likewise, analysis of robustness to model uncertainties is not discussed in this paper. Performing such analysis and clarifying the proposed method's sensitivity with respect to parameter errors is important because the sensitivity indicates which parameter needs to be correctly identified to obtain better control performance.}

We finally remark the limited applicability of the proposed prediction models for small-scale fish schools\del{,} although the control of small-scale schools is not within the scope of the current research. \del{Specifically, we have}\add{We} numerically observed the violation of the connectivity condition~$n_i>0$ in~\eqref{eq:OIneq0} \del{during the course of}\add{when}\del{tracking control for} \add{controlling} small-scale schools, typically for those with less than 20 individuals. We found that this phenomenon occurs because of \del{the }poor connectivity within small-scale schools. We leave \del{as an open problem }the development of a reduced-order prediction model without requiring \del{the }condition~\eqref{eq:OIneq0} \add{as an open problem}.

\section{Conclusion}

\del{In this paper, we have}\add{We} studied the MPC of fish schools. We \del{have }presented a reduced-order prediction model based on the weighted sum of directions with the weight being \add{the} normalized eigenvector centrality of the orientation-based attraction network in the fish school. We \del{have }then numerically confirmed the effectiveness of the MPC with the reduced-order model. The proposed approach \del{has }allowed the control of the fish school with hundreds of individuals\del{,} and \del{has }achieved smaller tracking errors than the MPC with the reduce-order model based on uniform weights. \del{One future direction worth pursuing is to consider more realistic models of fish schooling such as the one in~\cite{Filella2018} that can account for water flows. Another direction is to theoretically examine the generalizability of the current reduction approach to other swarm models.}

\add{There exist several future research directions that should be pursued. A research direction of practical importance is realizing reduced-order MPC with a guaranteed tracking performance, as such MPC is difficult to perform with the proposed prediction model. Another practically important research direction is generalization for the MPC of higher-order variables because the proposed method is limited to the control of the center of mass. Such a generalization would allow us to achieve complex fish behaviors such as milling.}
\add{More realistic models of fish schooling such as the one in~\cite{Filella2018} that can account for water flows also need to be considered. From the theoretical side, it is important to examine the generalizability of the current reduction approach to other swarm models.}

\section*{Acknowledgment}

This work is supported by JSPS KAKENHI Grant Number\add{s} JP21H01352\add{, JP21H01353, JP22H00514 and JST Moonshot R\&D Grant Number JPMJMS2284}. 

\appendix

\del{\section{\del{Proof of Lemma~{\ref{old:lem:2}}}}\label{app:}
\setcounter{section}{0}}

\del{Let $x_0\in\mathbb R^3$ be a  unit vector. We first consider the special case of~$x = x_0$. 
The Taylor series expansion of~$\phi$ shows the existence 
of a constant $C_{1}$ and an open set $\mathcal N \subset \mathbb R^3$ containing the origin but not the point $-x_0$ such that $\norm{\delta\phi(x_0, y)} \leq C_{1} \norm{y}^2$ for all $y \in \mathcal N$.
Define the sets
$\mathcal R = \{y \mid \norm{y} > 1/2\}$
and 
$\mathcal T = \mathbb R^3 \backslash (\mathcal N \cup \mathcal R)$.
If $y \in \mathcal R$, then $1\leq  2\norm{y}$ and, therefore, the triangle inequality shows $\norm{\delta\phi (x_0, y)} \leq 2(1+ \norm y) \leq 12 \norm{y}^2$. Also, because the mapping $y \mapsto \phi(x_0+y)$ is continuous over the compact set~$\mathcal T$, the maximum
$C_{2} = \max_{y \in \mathcal T} {\norm{\delta \phi(x_0, y)}}/{\norm{y}^2}$
exists and is finite. Therefore, the inequality \eqref{eq:f(x,y)bound} holds true if $C > \max(C_{1}, C_{2}, 12)$.}
 
\del{Let us then consider a general $x$. Take a constant $c > 0$ and an orthogonal matrix~{$M$} such that $x = c{M}x_0$. Because a simple calculation shows $\delta \phi(x, y) = Q\,\,\delta\phi(x_0, c^{-1}{M}^\top y)$ and~$c = \norm{x}$, we obtain 
    $\norm{\delta \phi(x, y)} 
    = \norm{\delta \phi(x_0, c^{-1}{M}^\top y)}
    < C \norm{c^{-1} {M}^\top y}^2
    = ({C}/{\norm{x}^2})\norm{y}^2$,  
which completes the proof.}

\add
\section{Proof of Theorem~\ref{thm:static}}\label{sec:prthmstatic}

For the ease of presentation, we omit the time variable~$k$ if no confusion is expected. We start by proving the following technical lemma on the normalization operator~$\phi$ defined in~\eqref{eq:def:phi}.

\renewcommand{\thesection}{\Alph{section}}
\begin{lemma}\label{lem:2}
{For all $x, y \in \mathbb R^3$ such that $x\neq 0$ and~$x+y\neq 0$, define }
\begin{equation}\label{eq:deltaphi}
    {\delta\phi(x, y) = \phi(x+y) - \left(
    \phi(x) + \frac{1}{\norm{x}} y - \frac{y^\top x}{\norm{x}^{3}}x 
    \right).}
\end{equation}
{Then,}
\begin{equation}\label{eq:f(x,y)bound}
{\norm{\delta \phi(x, y)} \leq \frac{2}{\norm x^2}\norm{y}^2 }
\end{equation}
{for all $x\neq 0$ and~$y \in \mathbb R^3 \backslash \{-x\}$.} 
\end{lemma}

\begin{proof}
Let us first consider the the special case of~$x = e_1$, where $e_1 = [1\ 0\ 0]^\top$ denotes the first element of the canonical basis of~$\mathbb R^3$. Let $y\in \mathbb R^3 \backslash \{-x\}$ be arbitrary. Define $z = y+x$ and let $z_1$, $z_2$, and~$z_3$ denote the first, second, and third entries of vector~$z$, respectively. Because $\norm{x} = \norm{e_1}=1$, $y = z-e_1$, and~$\delta \phi(x, y) = \phi(z) - [1\ z_2\ z_3]^\top$, to prove the lemma, we show
\begin{equation}\label{eq:ineqtoshow}
    \norm{\phi(z) - [1\ z_2\ z_3]^\top
    }^2\leq 4\norm{z-e_1}^4
\end{equation}
for all nonzero $z$. 

Let us write $R = \norm{z}$ and~$r = z_1$. Then, 
an algebraic manipulation shows 
\begin{equation}\label{eq:normphiz}
    \norm{\phi(z) - [1\ z_2\ z_3]^\top}^2 = 
    2+R^2-r^2-2\frac{R^2-r^2+r}{R}
\end{equation}
and 
\begin{equation}\label{eq:normze1}
    \norm{z-e_1}^4
    =
    (R^2-2r+1)^2. 
\end{equation}
Because $r$ can take any value within the interval $[-R, R]$, we introduce an auxiliary variable $s\in [-1, 1]$ and write $r=sR$. Then, using equations~\eqref{eq:normphiz} and~\eqref{eq:normze1}, we can show that the difference~$\Delta =  4\norm{z-e_1}^4 - \norm{\phi(z) - [1\ z_2\ z_3]^\top
    }^2$ admits the representation 
\begin{equation*}
    \Delta = a(R)s^2 + b(R)s + c(R), 
\end{equation*}
where 
   $a(R) = 17R^2-2R$, 
   $b(R) = -16R^3 - 16R+2$, and
   $c(R) = 4 R^4 + 3R^2 + 2$.
Therefore, to prove \eqref{eq:ineqtoshow}, we need to show that $\Delta \geq 0$ holds for all $R\geq 0$ and~$s \in [-1, 1]$. 
Let us first consider the case~$R\leq 1/4$. In this case, because $c(R)\geq 0$, we have $\Delta \geq  -\abs{a(R)}
-\abs{b(R)}
+c(R)$. We can show that function $R\mapsto {-\abs{a(R)}
-\abs{b(R)}
+c(R)}$ is nonnegative on the interval $[0, 1/4]$, and therefore, we conclude that $\Delta \geq 0$ holds for all $s$. 
Next, let us  consider the case~$R\geq 1/4$. In this case, we can show that $a(R)>0$ and~$b(R)^2-4a(R)c(R)\leq 0$. This shows that $\Delta$, as a quadratic function of~$s$, is nonnegative. This observation implies $\Delta\geq 0$ for the case~$R\geq 1/4$ as well. Hence, we showed that \eqref{eq:f(x,y)bound} holds if $x=e_1$ for all $y\neq e_1$

Let us consider a general $x$. Take an orthogonal matrix~{$M$} such that $x = \norm{x}{M}e_1$. A simple calculation shows $\delta \phi(x, y) = M\,\delta\phi(e_1, \norm{x}^{-1}{M}^\top y)$. Therefore, we can show 
\begin{equation*}
\begin{aligned}
    \norm{\delta \phi(x, y)} 
&= \Norm{\delta \phi(e_1, \norm{x}^{-1}{M}^\top y)}
\\
&\leq 2 \Norm{\norm{x}^{-1} {M}^\top y}^2
\\
&= ({2}/{\norm{x}^2})\norm{y}^2, 
\end{aligned}
\end{equation*}
as desired.
\end{proof}

We also prove the following lemma. 
\begin{lemma}\label{lem:deltaOi}
For all $i\in \{1, \dotsc, N\}$, the vector 
\begin{equation*}
    \delta_\alpha O_i = O_i - n_i \langle V\rangle _\alpha
\end{equation*}
satisfies
\begin{equation}\label{eq:deltaOiineq}
    \norm{\delta_\alpha O_i} \leq n_i \sqrt{Q_{i, n_i \alpha}}. 
\end{equation}
\end{lemma}

\begin{proof}
For all $i, j \in \{1, \dotsc, N\}$, define $\sigma_{ij} = 1 - V_i^\top V_j$. Definition \eqref{eq:def:O_i} of the vector~$O_i$ shows $\delta_\alpha O_i = \sum_{j=1}^N(\omega_{ij}-n_i\alpha_j)V_j$. Therefore, 
\begin{equation}\label{eq:thisequation}
    \norm{\delta_\alpha O_i}^2 
=
    -\sum_{j=1}^N \sum_{j'=1}^N 
     (\omega_{ij}-n_i\alpha_j) (\omega_{ij'}-n_i\alpha_{j'}) \sigma_{jj'}, 
\end{equation}
where we use equality $\sum_{j=1}^N (\omega_{ij} - n_i \alpha_j) = 0$. Therefore, equation~\eqref{eq:thisequation} and the trivial inequality $\sigma_{ij} \leq \theta_{ij}^2/2$ prove
$\norm{\delta_\alpha O_i}^2 
\leq 
\sum_{j=1}^N \sum_{j'=1}^N \abs{\omega_{ij'}-n_i\alpha_{j'}}\,
\abs{\omega_{ij}-n_i\alpha_j} \theta_{jj'}^2/2
=n_i^2 Q_{i, n_i\alpha}$, 
as desired.   
\end{proof}

Using Lemmas~\ref{lem:2} and~\ref{lem:deltaOi}, we prove the following preliminary error-bound for proving Theorem~\ref{thm:static}. 

\begin{proposition}\label{prop:zeta}
Let $i \in \{1, \dotsc, N\}$. {Assume that \eqref{eq:OIneq0} holds.} Define
\begin{equation}\label{eq:tildeD}
    w_i = \eta A_i+\xi_i u_i. 
\end{equation}
{If $\langle V \rangle_\alpha \neq 0$, then}\del{Then,} the norm of the vector \add{$\zeta_i$ defined by}
\begin{equation}\label{eq:zetai}
\zeta_i
= 
\phi(D_i) 
- 
\left(
\langle V\rangle_\alpha
+ 
\frac{w_i}{n_i}
-
\frac{w_i^\top \langle V\rangle_\alpha}{n_i\norm{\langle V\rangle_\alpha}^2}\langle V\rangle_\alpha
\right)
\end{equation}
satisfies 
\begin{equation*}
\begin{multlined}[.85\linewidth]
\norm{\zeta_i}
\leq 2\tilde \pi_\alpha \rho_i + 2 \rho_i^2 + \tilde \pi_i(1+4\rho_i) + (1 + 2\rho_i)\sqrt{Q_{i, n_i \alpha}}. 
\end{multlined}
\end{equation*}
\end{proposition}

\begin{proof}
\add{Let us first show that}
the vector~$\zeta_i$ \add{can be decomposed }as 
\begin{equation}\label{eq:zetaidecomp}
    \add{\zeta_i = \zeta_{i, 1} + \zeta_{i, 2} + \zeta_{i, 3} + \zeta_{i, 4} + \zeta_{i, 5}+ \zeta_{i, 6},}
\end{equation}
where vectors $\zeta_{i, 1}$, \dots, $\zeta_{i, 6}$ are defined by 
\begin{equation*}
    \begin{aligned}
    \zeta_{i, 1} &= \delta\phi(O_i, w_i), 
    \\
    \zeta_{i, 2} &= \phi(O_i) - \langle V \rangle_\alpha, 
    \\
    \zeta_{i, 3} &= ({\norm{O_i}^{-1}-n_i^{-1}}) w_i, 
    \\
    \zeta_{i, 4} &= - (\norm{O_i}^{-3}-n_i^{-3})(w_i^\top O_i) O_i, 
    \\
    \zeta_{i, 5} &= \frac{w_i^\top \langle V\rangle_\alpha}{n_i\norm{\langle V\rangle_\alpha}^2}\langle V\rangle_\alpha
    - 
    \frac{w_i^\top \langle V\rangle_\alpha}{n_i}\langle V\rangle_\alpha, 
    \\
    \zeta_{i, 6} &= \frac{w_i^\top \langle V\rangle_\alpha}{n_i}\langle V\rangle_\alpha
    - 
    \frac{w_i^\top O_i}{n_i^3}O_i. 
    \end{aligned}
\end{equation*}
\label{page:eq:zetaidecomp}By definition~\eqref{eq:deltaphi} of the mapping~$\delta\phi$ and identity~$D_i = O_i + w_i$, we can show $\delta\phi(O_i, w_i) + \phi(O_i)+\norm{O_i}^{-1}w_i - \norm{O_i}^{-3}(w_i^\top O_i)O_i = \phi(O_i + w_i) = \phi(D_i)$. Therefore, a simple algebraic manipulation yields
\begin{equation*}
    \zeta_{i, 1} + \zeta_{i, 2} + \zeta_{i, 3} + \zeta_{i, 4} = \phi(D_i) - \langle V\rangle_\alpha - n^{-1}w_i + n_i^{-3}(w_i^\top O_i) O_i. 
\end{equation*}
We can also easily see that
\begin{equation*}
\zeta_{i, 5} + \zeta_{i, 6} =
\frac{w_i^\top \langle V\rangle_\alpha}{n_i \norm{\langle V\rangle_\alpha}^2}\langle V\rangle_\alpha - \frac{w_i^\top O_i}{n_i^3}O_i. 
\end{equation*}
Hence, from \eqref{eq:zetai}, we can confirm the decomposition~\eqref{eq:zetaidecomp}, as required.

Next, let us evaluate the norm of each of these vectors. Lemma~\ref{lem:2} shows
\begin{equation}\label{eq:zeta1}
\begin{aligned}
    \norm{\zeta_{i, 1}} 
    &\leq 2 \norm{O_i}^{-2} \norm{w_i}^2
    =4 \rho_i^2. 
    \end{aligned}
\end{equation}
Because $\zeta_{i, 2} = (\norm{O_i}^{-1}-n_i^{-1})O_i - n_i^{-1}\delta_\alpha O_i$, inequality~\eqref{eq:deltaOiineq} shows
\begin{equation}\label{eq:isreferred_zeta2}
    \begin{aligned}
    \norm{\zeta_{i, 2}} 
    &\leq 
    (\norm{O_i}^{-1}-n_i^{-1})\norm{O_i} + \sqrt{Q_{i, n_i \alpha}}
    \\
    &=
    \tilde \pi_i + \sqrt{Q_{i, n_i \alpha}}. 
    \end{aligned}
\end{equation}
Similarly, we can show 
\begin{equation}\label{eq:zeta3}
\begin{aligned}
    \norm{\zeta_{i, 3}} 
    &\leq ({\norm{O_i}^{-1}-n_i^{-1}})\norm{w_i}
    = \tilde \pi_i \rho_i, 
    \end{aligned}
\end{equation}
and
\begin{equation}\label{eq:isreferred_zeta4}
\begin{aligned}
    \norm{\zeta_{i, 4}} 
    &\leq
    (\norm{O_i}^{-3}-n_i^{-3})\norm{w_i} \norm{O_i}^2
    \\
    &=
    \rho_i (1-\pi_i^3)
    \\
    &\leq 
    3 \tilde \pi_i \rho_i. 
\end{aligned}
\end{equation}
Furthermore, it follows that  
\begin{equation}\label{eq:isreferred_zeta5}
\begin{aligned}
    \norm{\zeta_{i, 5}}
    &\leq
    (1-\norm{\langle V\rangle_\alpha}^2)\norm{w_i}n_i^{-1}
    \\
    &\leq 2\tilde \pi_\alpha \rho_i. 
\end{aligned}
\end{equation}
By using the triangle inequality and inequality~\eqref{eq:deltaOiineq}, we can show 
\begin{equation}\label{eq:zeta6}
\begin{aligned}
    \norm{\zeta_{i, 6}} 
    &\leq  2\norm{w_i}n_i^{-1} \sqrt{Q_{i, n_i \alpha}}
    \\
    &\leq 
    2\rho_i \sqrt{Q_{i, n_i \alpha}}. 
\end{aligned}
\end{equation}
Finally, inequalities \cref{eq:zeta1,eq:isreferred_zeta2,eq:zeta3,eq:isreferred_zeta4,eq:isreferred_zeta5,eq:zeta6} complete the proof.
\end{proof}

We can now prove Theorem~\ref{thm:static}. Let  $w = A+u$. From definition \eqref{eq:tildeD} of vector~$w_i$, we have $\sum_{i=1}^N ({\alpha_i}/{n_i}) w_i = w$. Therefore, taking the weighted summation with respect to $i$ in equation~\eqref{eq:zetai} shows 
\begin{equation*}
\begin{aligned} 
   \langle \zeta\rangle_\alpha
    &=
    \begin{multlined}[t]
        \langle \phi(D)\rangle_\alpha 
    -
    \left(
    \langle V\rangle_\alpha + w - \frac{w^\top\langle V\rangle_\alpha}{\norm{\langle V\rangle_\alpha}^2}\langle V\rangle_\alpha
    \right)
    \end{multlined}
    \\
    &=
    \begin{multlined}[t]
    \langle \phi(D)\rangle_\alpha -
    \norm{\langle V\rangle_\alpha}\,\phi\bigl(\langle V\rangle_\alpha + w\bigr) + \norm{\langle V\rangle_\alpha}\, \delta\phi\bigl(\langle V\rangle_\alpha, w\bigr), 
    \end{multlined}
\end{aligned}
\end{equation*}
where in the last equation we used \eqref{eq:deltaphi}. Therefore,
\begin{equation}\label{eq:finalepailonalpha}
    \epsilon_\alpha = \langle \zeta\rangle_\alpha-
    \norm{\langle V\rangle_\alpha}\,\delta\phi\bigl(\langle V\rangle_\alpha, w\bigr). 
\end{equation}
Hence, the triangle inequality as well as 
Proposition~\ref{prop:zeta} and Lemma~\ref{lem:2} complete the proof of inequality~\eqref{eq:error:alpha}. 

\renewcommand{\thesection}{Appendix \Alph{section}}
\section{Proof of Theorem~\ref{thm:nec}}\label{sec:pfthmnec2}

Define vector~$\zeta_i$ as in \eqref{eq:zetai}. Assume that $\alpha$ equals the normalized eigenvector centrality of~$\mathcal G_o(k)$. Let us decompose vector~$\langle\zeta \rangle_\alpha = \sum_{i=1}^N \alpha_i \zeta_i$ as  $\langle\zeta \rangle_\alpha = \omega_{1} + \omega_{2} + \omega_{3} + \omega_{4} + \omega_{5} + \omega_6$, where $\omega_\ell = \sum_{i=1}^N \alpha_i \zeta_{i, \ell}$ for all $\ell=1,\dotsc, 6$. We bound the norm of  vectors $\omega_1$, $\omega_3$, $\omega_4$, $\omega_5$, and~$\omega_6$ using inequalities \eqref{eq:zeta1} and \eqref{eq:zeta3}--\eqref{eq:zeta6}. Let us evaluate $\norm{\omega_2}$ in a different manner. Because $\alpha$ is assumed to be equal to the normalized eigenvector centrality of~$\mathcal G_o(k)$, from equation \eqref{eq:intuition} we can show 
\begin{equation*}
\omega_2 
= 
\sum_{i=1}^N \alpha_i(\phi(O_i) - \langle V \rangle_\alpha) 
=  \sum_{i=1}^N \alpha_i (\phi(O_i) - n_i^{-1} O_i).
\end{equation*}
Thus, we obtain $\norm{\omega_2} \leq \sum_{i=1}^N \alpha_i \tilde \pi_i$. This inequality as well as inequalities~\eqref{eq:zeta1} and~\eqref{eq:zeta3}--\eqref{eq:zeta6} show
\begin{equation*}
\begin{multlined}[.85\linewidth]
\norm{\langle \zeta\rangle_\alpha}
 \leq \sum_{i=1}^N \alpha_i\bigg(2\tilde \pi_\alpha \rho_i + 2 \rho_i^2 
 + \tilde \pi_i(1+4\rho_i) + 2\rho_i\sqrt{Q_{i, n_i \alpha}}\bigg). 
 \end{multlined}
\end{equation*}
Therefore, equality~\eqref{eq:finalepailonalpha} as well as 
Proposition~\ref{prop:zeta} and Lemma~\ref{lem:2} complete the proof of the theorem, as desired. 

\color{black}






\end{document}